\documentclass[pra,superscriptaddress,nopacs,twocolumn]{revtex4-1}
\usepackage[utf8x]{inputenc}

\usepackage{xcolor}
\definecolor{darkred}{rgb}{0.8,0.1,0.1}
 \usepackage{amsmath,mathtools}
\usepackage{enumerate}
\usepackage{graphicx,epsfig,amsmath,amssymb,verbatim,color}
\usepackage{dsfont}
\usepackage{times}

\usepackage{hyperref}
\hypersetup{colorlinks=true,citecolor=blue,linkcolor=blue,filecolor=blue,urlcolor=blue,breaklinks=true}
\usepackage{cleveref}
\usepackage{graphicx}

\usepackage{url}
\newtheorem{definition}{Definition}
\newtheorem{proposition}{Proposition}

\newtheorem{theorem}[proposition]{Theorem}

\newtheorem{corollary}[proposition]{Corollary}

\def\squareforqed{\hbox{\rlap{$\sqcap$}$\sqcup$}}
\def\qed{\ifmmode\squareforqed\else{\unskip\nobreak\hfil
\penalty50\hskip1em\null\nobreak\hfil\squareforqed
\parfillskip=0pt\finalhyphendemerits=0\endgraf}\fi}
\def\endenv{\ifmmode\;\else{\unskip\nobreak\hfil
\penalty50\hskip1em\null\nobreak\hfil\;
\parfillskip=0pt\finalhyphendemerits=0\endgraf}\fi}
\newenvironment{proof}{\noindent \textbf{{Proof.}~}}{\hfill $\blacksquare$}

\newcommand{\nc}{\newcommand}
\nc{\rnc}{\renewcommand}
\nc{\beg}{\begin{equation}}
\nc{\eeq}{{\end{equation}}}
\nc{\beqa}{\begin{eqnarray}}
\nc{\eeqa}{\end{eqnarray}}
\nc{\lbar}[1]{\overline{#1}}
\nc{\bra}[1]{\langle#1|}
\nc{\ket}[1]{|#1\rangle}
\nc{\ketbra}[2]{|#1\rangle\!\langle#2|}
\nc{\braket}[2]{\langle#1|#2\rangle}

\nc{\proj}[1]{| #1\rangle\!\langle #1 |}
\nc{\avg}[1]{\langle#1\rangle}
\nc{\Rank}{\operatorname{Rank}}
\nc{\smfrac}[2]{\mbox{$\frac{#1}{#2}$}}
\nc{\tr}{\operatorname{Tr}}
\nc{\ox}{\otimes}
\nc{\dg}{\dagger}
\nc{\dn}{\downarrow}
\nc{\cA}{{\cal A}}
\nc{\cB}{{\cal B}}
\nc{\cC}{{\cal C}}
\nc{\cD}{{\cal D}}
\nc{\cE}{{\cal E}}
\nc{\cF}{{\cal F}}
\nc{\cG}{{\cal G}}
\nc{\cH}{{\cal H}}
\nc{\cI}{{\cal I}}
\nc{\cJ}{{\cal J}}
\nc{\cK}{{\cal K}}
\nc{\cL}{{\cal L}}
\nc{\cM}{{\cal M}}
\nc{\cN}{{\cal N}}
\nc{\cO}{{\cal O}}
\nc{\cP}{{\cal P}}
\nc{\cQ}{{\cal Q}}
\nc{\cR}{{\cal R}}
\nc{\cS}{{\cal S}}
\nc{\cT}{{\cal T}}
\nc{\cV}{{\cal V}}
\nc{\cX}{{\cal X}}
\nc{\cY}{{\cal Y}}
\nc{\cZ}{{\cal Z}}
\nc{\cW}{{\cal W}}
\nc{\csupp}{{\operatorname{csupp}}}
\nc{\qsupp}{{\operatorname{qsupp}}}
\nc{\var}{{\operatorname{var}}}
\nc{\rar}{\rightarrow}
\nc{\lrar}{\longrightarrow}
\nc{\polylog}{{\operatorname{polylog}}}
\nc{\1}{{\mathds{1}}}
\nc{\wt}{{\operatorname{wt}}}
\nc{\av}[1]{{\left\langle {#1} \right\rangle}}
\nc{\supp}{{\operatorname{supp}}}

\nc{\dia}{{\diamondsuit }}

\def\ve{\varepsilon}

\nc{\SSS}{{{\mathbb S}}}
\nc{\RR}{{{\mathbb R}}}
\nc{\CC}{{{\mathbb C}}}
\nc{\FF}{{{\mathbb F}}}
\nc{\NN}{{{\mathbb N}}}
\nc{\ZZ}{{{\mathbb Z}}}
\nc{\PP}{{{\mathbb P}}}
\nc{\QQ}{{{\mathbb Q}}}
\nc{\UU}{{{\mathbb U}}}
\nc{\EE}{{{\mathbb E}}}
\nc{\id}{{\operatorname{id}}}

\nc{\CHSH}{{\operatorname{CHSH}}}

\nc{\be}{\begin{equation}}
\nc{\ee}{{\end{equation}}}
\nc{\bea}{\begin{eqnarray}}
\nc{\eea}{\end{eqnarray}}
\nc{\<}{\langle}
\rnc{\>}{\rangle}
\nc{\Hom}[2]{\mbox{Hom}(\CC^{#1},\CC^{#2})}
\nc{\rU}{\mbox{U}}

\nc{\ob}[1]{#1}

\nc{\SEP}{{\text{SEP}}}
\nc{\NS}{{\text{NS}}}
\nc{\LOCC}{{\operatorname{LOCC}}}
\nc{\PPT}{{\operatorname{PPT}}}
\nc{\EXT}{{\text{EXT}}}
\nc{\Sym}{{\operatorname{Sym}}}

\nc{\HH}{\mathbb{H}}

\nc{\ERLO}{{E_{\text{r,LO}}}}
\nc{\ERLOCC}{{E_{\text{r,LOCC}}}}
\nc{\ERPPT}{{E_{\text{r,PPT}}}}
\nc{\ERLOCCinfty}{{E^{\infty}_{\text{r,LOCC}}}}
 
\nc{\sech}{\rm{sech}}

\rnc{\bar}{\;\rule{0pt}{9.5pt}\right|\;}

\nc{\lset}{\left\{\left.}
\nc{\rset}{\right\}}

\nc{\lsetr}{\left\{}
\nc{\rsetr}{\right.\right\}}
\nc{\barr}{\left|\rule{0pt}{9.5pt}\;}

\let\id\1

\nc{\norm}[2]{\left\lVert#1\right\rVert_{#2\!}}
\nc{\lnorm}[2]{\left\lVert#1\right\rVert_{\ell_{#2}}}

\nc{\STAB}{{{\operatorname{STAB}}}}
 \nc{\bu}{{{\textbf{u}}}}
  \nc{\sn}{{{\operatorname{sn}}}}

\begin{document}
 \title{Efficiently computable bounds for magic-state distillation}
 \author{Xin Wang}
\email{wangxin73@baidu.com}
\affiliation{Institute for Quantum Computing, Baidu Research, Beijing 100193, China}
\affiliation{Joint Center for Quantum Information and Computer Science, University of Maryland, College Park, Maryland 20742, USA}

 \author{Mark M. Wilde}
\email{mwilde@lsu.edu}
\affiliation{Hearne Institute for Theoretical Physics, Department of Physics and Astronomy,
Center for Computation and Technology, Louisiana State University, Baton Rouge, Louisiana 70803, USA}

 \author{Yuan Su}
 \email{buptsuyuan@gmail.com}
\affiliation{Joint Center for Quantum Information and Computer Science, University of Maryland, College Park, Maryland 20742, USA}
 \affiliation{Department of Computer Science and Institute for Advanced Computer Studies, University of Maryland, College Park, Maryland 20742, USA}

\date{\today}

\begin{abstract}
Magic-state distillation (or non-stabilizer state manipulation) is a crucial component in the leading approaches to realizing scalable, fault-tolerant, and universal quantum computation. Related to non-stabilizer state manipulation is the resource theory of non-stabilizer states, for which one of the goals is to characterize and quantify non-stabilizerness of a quantum state.
In this paper, we introduce the family of \textit{thauma} measures to quantify the amount of non-stabilizerness in a quantum state, and we exploit this family of measures to address several open questions in the resource theory of non-stabilizer states. 
As a first application, we establish the \textit{hypothesis testing thauma} as an efficiently computable benchmark for the one-shot distillable non-stabilizerness, which in turn leads to a variety of bounds on the rate at which non-stabilizerness can be distilled, as well as on the overhead of magic-state distillation.
We then prove that the \textit{max-thauma} can be used as an efficiently computable tool in benchmarking  the efficiency of magic-state distillation and that it can outperform pervious approaches based on mana.
Finally, we use the \textit{min-thauma} to bound a quantity known in the literature as the ``regularized relative entropy of magic.'' As a consequence of this bound, we find that two classes of states with maximal \textit{mana},  a previously established non-stabilizerness measure, cannot be interconverted in the asymptotic regime at a rate equal to one. This result resolves a basic question in the resource theory of non-stabilizer states and reveals a difference between the resource theory of non-stabilizer states and other resource theories such as entanglement and coherence.
\end{abstract}

\maketitle

\textbf{\textit{Introduction.}}---Quantum computers hold the promise of a substantial speed-up over classical computers for certain algebraic problems \cite{Shor1997,Grover1996,Childs2010} and the simulation of quantum systems~\cite{Lloyd1996,Childs2018a}. One  main obstacle to the physical realization of quantum computation is the decoherence that occurs during the execution of quantum algorithms. 
Fault-tolerant quantum computation (FTQC) \cite{Shor1996,Campbell2017c} provides a framework to overcome this difficulty and allows reliable quantum computation when the physical error rate is below a certain threshold value.

According to the Gottesman--Knill theorem \cite{Gottesman1997,Aaronson2004a}, a quantum circuit comprised of only Clifford gates confers no computational advantage because it can be simulated efficiently on a classical computer. 
However, the addition of a  \textit{non-stabilizer state} can lead to a universal gate set via a technique called state injection \cite{Gottesman1999,Zhou2000}, thus achieving universal quantum computation. 
The key of this resolution is to perform \textit{magic-state distillation} \cite{Bravyi2005} (see \cite{Bravyi2012,Jones2013,Haah2017,Campbell2018,Hastings2018,Krishna2018a,Chamberland2018} for recent progress), wherein stabilizer operations are used to transform a large number of noisy non-stabilizer states into a small number of high quality non-stabilizer states. Therefore, a quantitative theory is
highly desirable in order to fully exploit the power of non-stabilizer states in fault-tolerant quantum computation.

Quantum resource theories (QRTs) offer a powerful framework for studying different phenomena in quantum physics, and the seminal ideas of QRTs have recently been influencing diverse areas of physics~\cite{Chitambar2018}.
In the context of the non-stabilizer-state model of universal quantum computation, the resource-theoretic approach reduces to the characterization and quantification of the usefulness of the resourceful non-stabilizer states \cite{Veitch2012,Veitch2014}.
In the framework of \cite{Veitch2012,Veitch2014}, the free operations are the stabilizer operations, those that possess a fault-tolerant implementation in the context of fault-tolerant quantum computation, and the free states are the stabilizer states (STAB). Stabilizer operations include preparation and measurement in the computational basis, and a restricted set of unitary operations, called the Clifford unitaries.
The free states consist of all pure stabilizer states, which are eigenstates of the generalized Pauli operators, and their convex mixtures.
The resource states, namely, the non-stabilizer states, are key resources that are required to achieve some desired computational tasks.
For quantum computers acting on qudit registers with
odd dimension $d$, the resource theory of non-stabilizer states (or equivalently contextuality with respect to stabilizer measurements \cite{Howard2014,Delfosse2015}) has been developed \cite{Veitch2012,Mari2012,Veitch2014}. The resource theory of non-stabilizer states for multiqubit systems was recently developed in \cite{Howard2016,Bravyi2016,Heinrich2018}.

In this paper, we solve some fundamental open questions in the resource theory of non-stabilizer states, and we develop the framework for one-shot magic state distillation. Our main tool for doing so is the \textit{thauma} family of non-stabilizer monotones, which quantify the amount of non-stabilizerness in a given state by comparing it to a positive semi-definite operator with non-positive mana (i.e., a subnormalized state with no non-stabilizerness). Our first contribution is to introduce the one-shot distillable non-stabilizerness of a quantum state and an upper bound for it named \textit{hypothesis testing thauma}. This result leads to various applications for magic-state distillation, which can be interpreted as fundamental limits. The \textit{max-thauma} is another member of the thauma family, and we prove that it is an efficiently computable non-stabilizerness monotone, which can in turn be used to evaluate the efficiency of magic-state distillation. We further provide an example to demonstrate that max-thauma outperforms mana in benchmarking the efficiency of magic-state distillation. We also prove that the \textit{min-thauma} is an additive lower bound on the ``regularized relative entropy of magic,'' the latter quantity defined in \cite{Veitch2014}. This bound then leads to the conclusion that two magic states with maximal negativity cannot be interconverted asymptotically at a rate equal to one.

\textbf{\textit{Discrete Wigner function.}}---We now recall the definition of the discrete Wigner function \cite{Wootters1987,Gross2006,Gross2007}, which is an essential tool in the analysis of the resource theory of non-stabilizer states.
Throughout this paper, a Hilbert space implicitly has an odd dimension, and if the dimension is not prime, it should be understood to be a tensor product of Hilbert spaces each having odd prime dimension.

Let $\cH_d$ denote a Hilbert space of dimension $d$, and let $\{\ket j\}_{j=0,\ldots,d-1}$ denote the standard computational basis. For a prime number~$d$, we define the respective shift and boost  operators $X,Z\in\cL(\cH_d)$ as
$X\ket j = \ket{j\oplus 1}$ and
$Z\ket j =\omega^j \ket j$, with $\omega=e^{2\pi i /d}$.
We define the Heisenberg--Weyl operators as
$T_{\bu}= \tau^{-a_1a_2}Z^{a_1}X^{a_2}$,
where $\tau=e^{(d+1)\pi i/d}$ and $\bu=(a_1,a_2)\in \ZZ_d\times \ZZ_d$.

 For each point $\bu$ in the discrete phase space, there is a corresponding operator $A_\bu$, and the value of the discrete Wigner representation of a quantum state $\rho$ at this point is given by
$
W_{\rho}(\bu):=  \tr A_\bu\rho/d, 
$
where $\{A_\bu\}_\bu$ are the phase-space point operators:
$
A_\bu:=T_\bu A_0 T_\bu^\dag$, $ A_0:=\frac1d \sum_\bu T_\bu$.
We give more details of this formalism in Appendix~A \footnote{See Supplemental Material [url] for detailed mathematical developments of the assertions in the main text. The Supplemental Material includes Refs.~\cite{PhysRevA.71.042315,Ferrie2011,Fawzi2017a,Tomamichel2009,Mosonyi2011,Uhlmann1976,FL13,Beigi2013,Wilde2018a}}.

\textbf{\textit{Thauma.}}---It is well known that quantum computations are classically simulable if they consist of stabilizer operations acting on quantum states with a positive discrete Wigner function. Such states are thus useless for magic-state distillation \cite{Veitch2012}.
Let $\cW_+$ denote the set of quantum states with positive discrete Wigner function.
States in $\cW_+$ can be understood as being analogous to states with a positive partial transpose  in entanglement distillation \cite{Peres1996,Horodecki1998}, in the sense that such states are undistillable. 

To address open questions in the resource theory of non-stabilizer states, we are motivated by the idea of the Rains bound from entanglement theory \cite{Rains2001}, as well as its variants~\cite{Wang2016,Wang2017e,Fang2017}, which also have fruitful applications in quantum communication~\cite{Tomamichel2015a,Tomamichel2016,Wang2016a,Wang2017d}. As developed in \cite{Rains2001} and the later work \cite{Audenaert2002}, the Rains bound and its variants consider sub-normalized states with non-positive logarithmic negativity~\cite{Vidal2002,Plenio2005b} as useless resources, and they use the divergence between the given state and such sub-normalized states to evaluate the behavior of entanglement distillation.

Thus, inspired by the main idea behind the Rains bound, we 
introduce the set of sub-normalized states with non-positive mana:
$
	\cW:=\{\sigma:  \cM(\sigma)\leq 0,\  \sigma \ge 0\},
$ with the mana $\cM(\rho)$ of a quantum state $\rho$ defined as \cite{Veitch2014}
\begin{align}
	\cM(\rho):=\log_2 \Vert \rho \Vert_{W,1},
	\notag
\end{align}
where the Wigner trace norm of an operator $V$ is defined as $\Vert V \Vert_{W,1} := \sum_{\bu} |W_{V}(\bu)|$.
It follows from definitions that $\tr \sigma \leq 1$ if $\sigma \in \cW$.
Note that the mana \cite{Veitch2014} is analogous to the logarithmic negativity \cite{Vidal2002,Plenio2005b}.
Furthermore, the following strict inclusions hold: $\text{STAB} \subsetneq \cW_+ \subsetneq \cW$.

We now define the \textit{thauma} \footnote{Greek for ``wonder'' or ``marvel''} of a state $\rho$ as 
\begin{align}
	\theta(\rho):=\min_{\sigma\in \cW} D(\rho \| \sigma),
	\notag
	\end{align}
where $D(\rho \| \sigma)$ is the quantum relative entropy \cite{Umegaki1962}, defined as $D(\rho \| \sigma)=\tr\{\rho[\log_2 \rho - \log_2 \sigma]\}$ when the support of $\rho$ is contained in the support of $\sigma$ and equal to $+\infty$ otherwise. The thauma can be understood as the minimum relative entropy between a quantum state and the set of subnormalized states with non-positive mana. The thauma is a non-stabilizerness measure that can be efficiently computed via convex optimization (see Appendix~B \cite{Note1}).
Following from the definition of thauma above, we define the regularized thauma of a state $\rho$  as 
$	\theta^{\infty}(\rho):=\lim_{n\to \infty}\theta(\rho^{\ox n})/n$.

The definition of thauma given above can be generalized to a whole family of thauma measures of non-stabilizerness. Defining a generalized divergence $\mathbf{D}(\rho\|\sigma)$ to be any function of a quantum state $\rho$ and a positive semi-definite operator $\sigma$ that obeys data processing \cite{Polyanskiy2010b,SW12}, i.e., $\mathbf{D}(\rho\|\sigma) \geq \mathbf{D}(\cN(\rho)\|\cN(\sigma))$ where $\cN$ is a quantum channel, we arrive at the generalized thauma of a quantum state $\rho$:
\begin{align}
	\boldsymbol{\theta}(\rho):=\inf_{\sigma\in \cW} \mathbf{D}(\rho \| \sigma).
	\notag
\end{align}
If the generalized divergence $\mathbf{D}$ is non-negative for a state $\rho$ and a sub-normalized state $\sigma$ and equal to zero if $\rho = \sigma$, then it trivially follows that the generalized thauma $\boldsymbol{\theta}(\rho)$ is a non-stabilizerness monotone, meaning that it is non-increasing under the free operations and equal to zero for stabilizer states. Examples of generalized divergences, in addition to the relative entropy, include the Petz--R\'enyi relative entropies \cite{P86} and the sandwiched R\'enyi relative entropies \cite{Muller-Lennert2013,Wilde2014a}. See Appendix~C for further details \cite{Note1}.

\textbf{\textit{Min- and max-thauma.}}---In what follows, we make use of the Petz--R\'enyi relative entropy of order zero \cite{P86} and the max-relative entropy \cite{Datta2009} to define the min-thauma and the max-thauma, respectively. As we prove in what follows, these two members of the thauma family are efficiently computable by semidefinite programs (SDPs) \cite{Boyd2004} and are particularly useful for addressing open questions in the resource theory of non-stabilizer states.

The min-thauma of a state $\rho$ is defined as 
	\begin{align}
	\theta_{\min}(\rho)  :=\min_{\sigma\in \cW} D_0(\rho \| \sigma)
	:= \min_{\sigma\in \cW} [-\log_2  \tr P_\rho\sigma],
	\notag
	\end{align}
	where $P_\rho$ denotes the projection onto the support of $\rho$. Note that $\theta_{\min}(\rho)$ is an SDP and the duality theory of SDPs \cite{Boyd2004} leads to  the dual SDP:
\begin{align}
\theta_{\min}(\rho)=-\log_2 \min\{\Vert Q \Vert_{W,\infty}: Q\ge P_\rho \},
\notag
\end{align}
where $\Vert V \Vert_{W,\infty} := d \max_{\bu} |W_{V}(\bu)|$ denotes the Wigner spectral norm of an operator $V$ acting on a space of dimension~$d$.
For any pure state $\ket \psi$,
\begin{align}
\theta_{\min}(\psi)=-\log_2\max_{\sigma\in\cW} F(\psi,\sigma)\le -\log_2 F_{\text{Stab}}(\psi).
\notag
\end{align} 
where $F_{\text{Stab}}(\psi)$ is the stabilizer fidelity \cite{Bravyi2018}.

The max-thauma of a state $\rho$ is defined as 
	\begin{align}
	\theta_{\max}(\rho) & :=\min_{\sigma\in \cW} D_{\max}(\rho \| \sigma) 
	 := \min_{\sigma\in \cW} \left[\min \{\lambda: \rho \leq 2^{\lambda} \sigma\}\right] \notag \\
	&=\log_2 \min\left\{ \Vert V \Vert_{W,1} 
	:  \rho \le V \right\}. \notag
	\end{align}

As the following proposition states, the min- and max-thauma are additive non-stabilizerness measures. Additionally, the min-thauma is a lower bound for the regularized thauma, and the max-thauma is an upper bound. 

\begin{proposition}
\label{prop: min theta add}
For states $\rho$ and $\tau$, it holds that 
\begin{align}
\theta_{\min}(\rho\ox\tau) & =\theta_{\min}(\rho)+\theta_{\min}(\tau),\notag \\
\theta_{\max}(\rho\ox\tau) & =\theta_{\max}(\rho)+\theta_{\max}(\tau) \notag.
\end{align}
Consequently,
$\theta_{\min}(\rho) \leq \theta^{\infty}(\rho)\leq \theta_{\max}(\rho) $.
\end{proposition}

The proof of Proposition~\ref{prop: min theta add} relies on the Petz--R\'enyi relative entropy of order zero \cite{Datta2009}, the max-relative entropy \cite{Datta2009}, and the duality theory of SDPs \cite{Boyd2004} (see Appendix~D for details \cite{Note1}).

In Appendix~E \cite{Note1}, we prove that the max-thauma possesses a stronger monotonicity property, in the sense that it does not increase on average under stabilizer operations.

We note here that an important consequence of the additivity of min-thauma is that the maximum overlap between $\ket {\phi}^{\ox n}$ and the set $\cW$ is $2^{-n\theta_{\min}(\phi)}$; i.e., for any   $\tau \in \cW \ (\text{or STAB})$, we have that
$
    \tr [ \proj {\phi}^{\ox n} \tau ] \leq 2^{-n\theta_{\min}(\phi)}
$.

\textbf{\textit{Thauma for basic non-stabilizer states.}}---Proposition~\ref{prop: thauma collapse magic states} below states that the min-, regularized, and max-thauma collapse to the same value for several interesting non-stabilizer states, including the Strange, Norrell, $H$, and $T$ state.

The Strange and Norrell states are defined as  \cite{Veitch2014}
 \begin{align}
 	\ket \SSS & :=  (\ket 1-\ket 2)/\sqrt 2, \qquad
	\ket \NN :=  (-\ket 0+2\ket 1-\ket 2)/\sqrt 6. \notag
 \end{align}

The qutrit Hadamard gate is given by \cite{Howard2012}
\begin{align}
H=\frac{1}{\sqrt{3}}
\left(\begin{matrix}
1 &1& 1\\
1 & \omega & \omega^2\\
1 & \omega^2 & \omega 
\end{matrix}
\right),
\label{eq:qutrit-hadamard}
\end{align}
where we recall that $\omega = e^{2 \pi i/3}$.
The $H$ gate has eigenvalues $+1$, $-1$, and $i$, and we label the  three corresponding eigenstates as $\ket {H_+}$, $\ket {H_-}$, and $\ket {H_i}$. The $\ket {H_+}$ state is a non-stabilizer state that is typically considered in the context of magic-state distillation \cite{Bravyi2005,Anwar2012}. In what follows, we refer to it as the $H_+$ non-stabilizer state.

Another common choice for a non-Clifford gate is the $T$ gate. The qutrit $T$ gate is given by
$
T:=\text{diag}(\xi,1,\xi^{-1})
$,
where $\xi=e^{2\pi i/9}$ is a primitive ninth root of unity \cite{Howard2012}.
The non-stabilizer state corresponding  to the qutrit $T$ gate is 
$
\ket T := \frac{1}{\sqrt 3}(\xi  \ket  0+ \ket 1 + \xi^{-1}\ket 2), 
$
which is the state resulting from applying the $T$ gate to the stabilizer state $(  \ket  0+ \ket 1 + \ket 2)/\sqrt 3$.

 In what follows, we employ the shorthand $\SSS \equiv \proj{\SSS} $, $\NN \equiv \proj{\NN} $, $H_+ \equiv \proj{H_+}$, and $T\equiv \proj{T}$. Before stating the theorem, let us recall the definition of the ``regularized relative entropy of magic'' and the ``relative entropy of magic'' \cite{Veitch2014}:\begin{align}
R_{\cM}^{\infty}({\rho})  & := \lim_{n \to \infty}\frac{1}{n} R_{\cM}({\rho}^{\otimes n}),\quad
R_{\cM}({\rho})  := \min_{\sigma \in \text{STAB}}D(\rho \| \sigma).\notag
\end{align}

\begin{proposition}
\label{prop: thauma collapse magic states}
The following equalities hold
\begin{align}
\theta_{\min}(\SSS) & = \theta^{\infty}(\SSS) = \theta_{\max}(\SSS) = \log_2(5/3), \notag\\
\theta_{\min}(\NN) & = \theta^{\infty}(\NN) = \theta_{\max}(\NN) = R_{\cM}^{\infty}(\NN) =  \log_2(3/2) ,\notag\\
\theta_{\min}(H_+) & = \theta^{\infty}(H_+) = \theta_{\max}(H_+) \notag \\
& = R_{\cM}^{\infty}(H_+) = \log_2(3-\sqrt 3), \notag \\
\theta_{\min}(T) & = \theta^{\infty}(T) = \theta_{\max}(T) = \log_2(1+2\sin(\pi/18)).
\notag
\end{align}
\end{proposition}

Appendix~F \cite{Note1} contains a proof of Proposition~\ref{prop: thauma collapse magic states}.
In the forthcoming sections, we provide applications of Propositions~\ref{prop: min theta add} and \ref{prop: thauma collapse magic states} to the resource theory of non-stabilizer states.



\textbf{\textit{Fundamental limits for distilling non-stabilizer states.}}---The basic task of magic-state distillation \cite{Bravyi2005} can be understood as follows.  For any given quantum state $\rho$, we aim to transform this state to a collection of non-stabilizer states (e.g., $\ket T$) with high fidelity using stabilizer operations. 
The goal is to maximize the number of target states while keeping the transformation infidelity within some tolerance $\ve$. After magic-state distillation, one can use a circuit gadget (which requires only stabilizer operations) to transform this non-stabilizer state into a non-Clifford gate \cite{Gottesman1999,Zhou2000}. Protocols for distillation in the qudit setting of quantum computing were recently developed in \cite{Campbell2012b,Anwar2012,Campbell2014,Dawkins2015}.

 In the following, we study the fundamental limit of magic-state distillation of a pure target non-stabilizer state. We define the approximate one-shot distillable $\phi$-non-stabilizerness of a given state $\rho$ as the maximum number of $\proj {\phi}$ non-stabilizer states that can be obtained via stabilizer operations, while keeping the infidelity within a given tolerance. 
 Formally, for any triplet $(\rho,\phi,\ve)$ consisting of an initial state $\rho$, a target pure state $\phi$, and an infidelity tolerance $\ve$, 
the one-shot $\ve$-error distillable $\phi$-non-stabilizerness of $\rho$ is defined
to be the maximum number of $\phi$ non-stabilizer states achievable via stabilizer operations, with an error tolerance of $\ve$:
\begin{equation}
\cM_{\phi}^\ve(\rho) = \sup \{ k: \Lambda(\rho)\approx_\ve \proj {\phi}^{\ox k}, \ \Lambda \in \text{SO} \},
\notag
\end{equation}
where $\proj{\psi} \approx_\ve \sigma$ is a shorthand for $\bra{\psi}\sigma\ket{\psi}\geq 1 - \varepsilon$ and SO for stabilizer operations.

  In what follows, we focus on the one-shot distillable $H_+$-non-stabilizerness $\cM_{H_+}^\ve(\rho)$ and the one-shot distillable $T$-non-stabilizerness $\cM_{T}^\ve(\rho)$.
 
We first connect the task of magic-state distillation to quantum hypothesis testing between non-stabilizer states and operators in the set $\cW$ (recall that $\text{STAB}\subsetneq \cW$), and  we note here that such an approach was previously taken in  entanglement theory \cite{Rains1999,Rains2001}. Quantum hypothesis testing is the task of distinguishing two possible states $\rho_0$ and $\rho_1$ (null hypothesis $\rho_0$, alternative hypothesis $\rho_1$). We are allowed to perform a measurement characterized by the POVM $\{M, \1-M\}$ with respective outcomes $0$ and~$1$. If the outcome is $0$, we accept the null hypothesis. Otherwise, we accept the alternative one. The probabilities of \emph{type-I} and \emph{type-II} errors are given by $\tr (\1-M) \rho_0$ and $\tr M \rho_1$, respectively. The hypothesis testing relative entropy \cite{Buscemi2010,Wang2012} quantifies the minimum type-II error probability provided that the type-I error probability is within a given tolerance: 
$D_H^\ve(\rho_0\|\rho_1) :=-\log_2\min \big\{ \tr M\rho_1 \,\big|\, 0\le M\le \1,\, 1-\tr M\rho_0\le\ve \big\}$.

\begin{proposition}
\label{th:hypothesis bound magic}
Given a state $\rho$, the following holds
\begin{align}
\cM_{H_+}^\ve(\rho) & \le \frac{\min_{\sigma\in \cW}D_H^\ve(\rho\|\sigma)}{\log_2(3-\sqrt3)},\label{eq:H-magic-bnd}\\
\cM_{T}^\ve(\rho) & \le \frac{\min_{\sigma\in \cW}D_H^\ve(\rho\|\sigma)}{\log_2(1+2\sin(\pi/18))}.
\label{eq:T-magic-bnd}
\end{align}
\end{proposition}

A consequence of Proposition~\ref{th:hypothesis bound magic} is that the thauma of a quantum state is an upper bound on its distillable $H_+$ (or $T$) non-stabilizerness.
Specifically, by applying the quantum Stein's lemma \cite{Hiai1991,ogawa2000strong,Tomamichel2013a,Li2014a}, we find the following:
\begin{corollary}
\label{coro:second order theorem}
The distillable non-stabilizerness of $\rho$ satisfies
\begin{align*}
\cM_{H_+}\!(\rho)=\lim_{\ve \to 0}\lim_{n\to \infty} \frac1n\cM_{H_+}^\ve(\rho^{\ox n}) & \le \frac{\theta(\rho)}{\log_2(3-\sqrt3)}, \notag\\
\cM_{T}(\rho) =\lim_{\ve \to 0}\lim_{n\to \infty} \frac1n\cM_T^\ve(\rho^{\ox n}) & \le \frac{\theta(\rho)}{\log_2(1+2\sin(\pi/18))}.
\notag
\end{align*}
\end{corollary}


\begin{figure}
\begin{center}
\includegraphics[width=6.2cm]{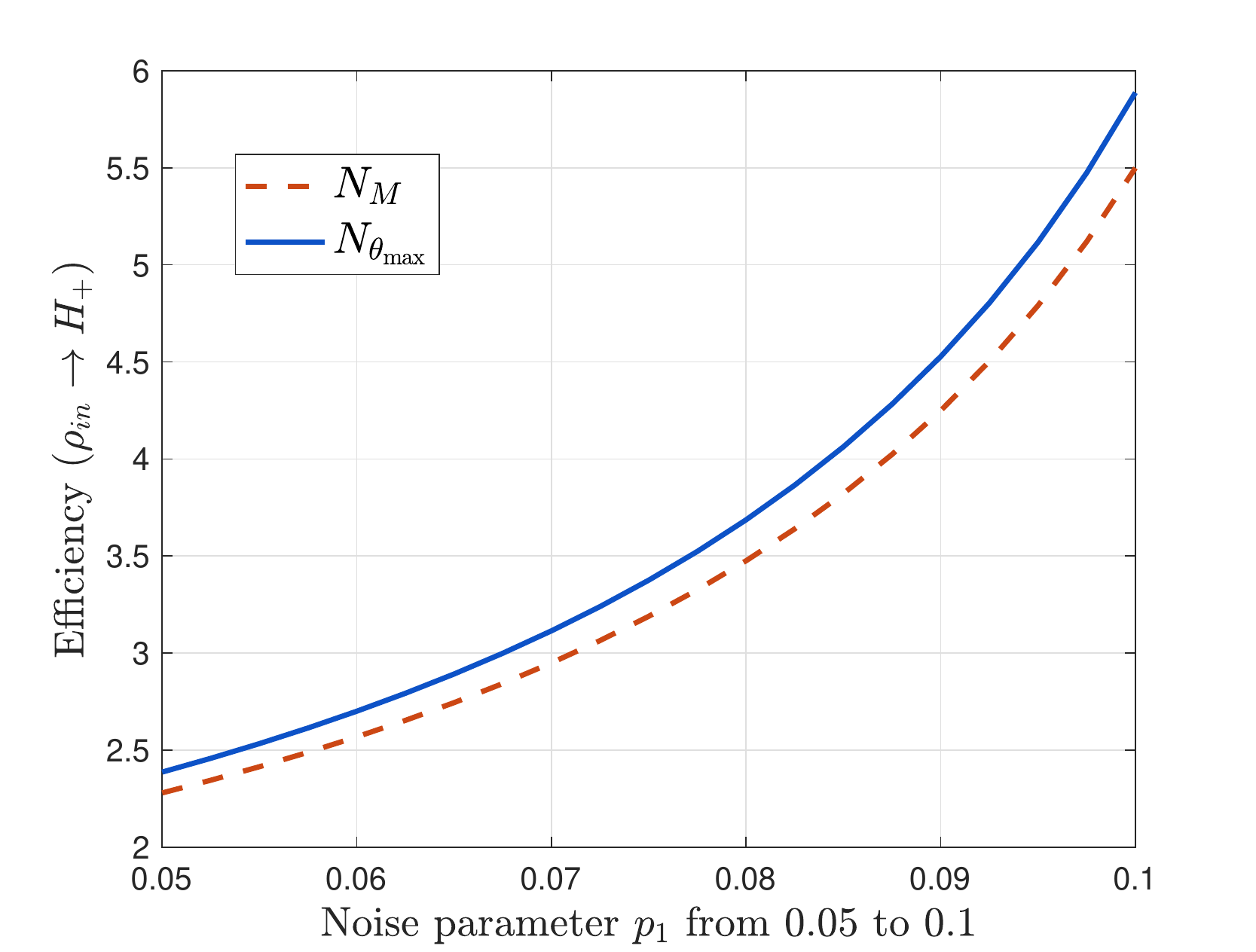}
\caption{Comparison between $N_{\theta_{\max}}(\rho_{\text{in}}\to H_+)$ and $N_\cM(\rho_{\text{in}}\to H_+)$ for input $\rho_{\text{in}}=(1-p_1-p_2)\proj{H_+}+p_1\proj{H_-}+p_2\proj{H_i}$ with $p_2 = 1/10$.}
\label{fig: ef 1}
\end{center}
\end{figure}

\textbf{\textit{Efficiency of magic-state distillation.}}---The efficiency of distilling a non-stabilizer state $\xi$ from several independent copies of a resource state $\rho$ is given by the minimum number of copies of $\rho$ needed, on average, to produce $\xi$ using stabilizer operations:
\begin{equation}
N_{\text{eff}}(\rho \to \xi)=
\inf\left\{n/p: \Lambda(\rho^{\ox n})\to \xi \text{ w/ prob. } p,\ \Lambda\in \text{SO}\right\}.
\notag
\end{equation}
 
Previously, the authors of \cite{Veitch2014} derived the following lower bound on the efficiency of magic-state distillation:
\begin{equation}\label{eq:lower-bnd-efficiency-veitch}
N_{\text{eff}}(\rho \to \xi)
\ge N_{\cM}(\rho,\xi):=\cM(\xi)/\cM(\rho).
\end{equation}
The lower bound in \cite{Veitch2014} was established by employing the mana of non-stabilizer states. Here, we utilize similar ideas and show that the max-thauma can also be applied to bound the efficiency of magic-state distillation.
\begin{proposition}
\label{th:eff bound}
The efficiency of distilling a non-stabilizer state $\xi$ from resource states $\rho$ is lower bounded by
$
N_{\theta_{\max}}(\rho,\xi):=\theta_{\max}(\xi)/\theta_{\max}(\rho).
$
\end{proposition}

Figure~\ref{fig: ef 1} demonstrates that the lower bound from Proposition~\ref{th:eff bound} can be tighter than the lower bound in~\eqref{eq:lower-bnd-efficiency-veitch}, thus giving an improved sense of the efficiency.

\textit{\textbf{On the overhead of magic-state distillation.}}---The overhead of magic-state distillation is defined as the ratio of the number of
input to output states, under a target error rate \cite{Bravyi2012,Krishna2018a}. Although our notion of error for magic-state distillation is different from that typically employed in the literature, we note here that the inverse of the one-shot distillable $\phi$-non-stabilizerness (i.e.,  $[\cM_{\phi}^\varepsilon(\rho)]^{-1}$) can be considered a reasonable way to measure the overhead of magic-state distillation. Then our upper bounds in Proposition~\ref{th:hypothesis bound magic} 
and Corollary~\ref{coro:second order theorem} become lower bounds on the overhead.

\textbf{\textit{Inequivalence between non-stabilizer states with maximal mana.}}---A fundamental problem in any quantum resource theory is to determine the interconversion rate between different resource states \cite{Chitambar2018}, in particular, between given states and maximally resourceful states. This is rooted in the fact that in any resource theory maximally resourceful states play a unique role in quantifying the resourcefulness of other states and assessing the performance of resource manipulation. Considering entanglement theory (or coherence theory) as an example, the interconversion between a given state and maximally entangled (coherent) states leads to fundamental tasks such as entanglement (coherence) distillation and dilution~\cite{Horodecki2009a,Streltsov2016,Chitambar2018}. Notably, any two maximally entangled (coherent) states under all resource measures in the same dimension are equivalent under free operations.

However, this is not the case in the resource theory of non-stabilizer states. Surprisingly, we find that even though the Strange state and the Norrell state each have maximum mana and are thus the most costly resource to simulate on a classical computer using certain known algorithms~\cite{Veitch2014,PWB15}, they are not equivalent even in the asymptotic regime. Note that the mana plays a significant role as a measure of non-classical resources in quantum computation~\cite{Mari2012,PWB15}. In particular, recall that mana is a non-stabilizerness measure analogous to logarithmic negativity in entanglement theory. In contrast,  logarithmic negativity of a bipartite state is equal to its maximal value if and only if the state is maximally entangled.

To establish this result, we recall that the asymptotic conversion rate from $\rho$ to $\xi$ under asymptotically-non-stabilizer-non-generating transformations is given by the ratio of their regularized relative entropies of resource~\cite{Brandao2015a}. That is,
$
R(\rho\to\xi)={R_{\cM}^{\infty}({\rho}) }/{R_{\cM}^{\infty}( {\xi})}.
$
We further recall that the Strange and Norrell 
 states have maximum mana \cite{Veitch2014}:
$
\cM(\SSS)=\cM(\NN)=\log_2 (5/3).
$
However, Proposition~\ref{prop: thauma collapse magic states} and the fact that $R_{\cM}^{\infty}( {\SSS}) \ge \theta^{\infty}( {\SSS})$ indicate that there is a gap between their ``regularized relative entropies of magic.''
As a consequence, we find that

\begin{theorem}
\label{th: S inequal N}
For the Strange state $\ket {\SSS}$ and the Norrell state $\ket {\NN}$, the following holds
\begin{align}
R(\NN\to \SSS)=R_{\cM}^{\infty}( {\NN})/R_{\cM}^{\infty}( {\SSS})\le \frac {\log_2 (3/2)}{\log_2 (5/3)}<1. 
\notag
\end{align}
\end{theorem}

Since stabilizer operations are included in the set of asymptotically-non-stabilizer-non-generating transformations, this result also establishes that the rate to obtain the Strange state from the Norrell state is smaller than one under stabilizer operations. 
Thus, the gap between $R_{\cM}^{\infty}( {\NN})$ and $R_{\cM}^{\infty}( {\SSS})$, as established in Theorem~\ref{th: S inequal N}, closes an open question from \cite[Section 4.4]{Veitch2014}.

This result demonstrates a fundamental difference between the resource theory of non-stabilizer states and the resource theory of entanglement or coherence. Specifically, we show that the maximally resourceful non-stabilizer states under certain resource measure cannot be interconverted at a rate equal to one, even in the asymptotic regime, while the maximally resourceful states in entanglement theory or coherence theory can be interconverted equivalently in the one-copy setting. However, it remains open to determine whether the conversion rate from the Strange state to the Norrell state is strictly smaller than $\log_2 (5/3)/\log_2 (3/2)$. 
Such an inequality would imply the irreversibility of asymptotic magic state manipulation.

\textbf{\textit{Conclusions.}}---We have introduced the thauma family of measures to quantify and characterize the non-stabilizerness resource possessed by quantum states that are needed for universal quantum computation. The min- and max-thauma are efficiently computable by semi-definite programming and lead to bounds on the rates at which one can interconvert non-stabilizer states. These bounds have helped to solve pressing open questions in the resource theory of non-stabilizer states.
More generally, our work establishes fundamental limitations to the processing of quantum non-stabilizerness, opening new perspectives for its investigation and exploitation as a resource in quantum information processing and quantum technology. Along this line, we suspect that our results will have immediate impact on the quantum optics community working on the resource theory of non-Gaussianity \cite{ZSS18,TZ18,AGPF18} and continuous-variable quantum computing \cite{LB99,GKP01}, because the main idea behind the thauma measure can be generalized to this setting.

\begin{acknowledgments}
We thank Ali Hamed Mosavian and Victor Veitch for discussions regarding \cite{Veitch2014}. XW
acknowledges support from the Department of Defense. MMW acknowledges support from NSF under grant no.~1714215. YS acknowledges support from the Army Research Office (MURI award W911NF-16-1-0349), the Canadian Institute for Advanced Research, the National Science Foundation (grants 1526380 and 1813814), and the U.S.\ Department of Energy, Office of Science, Office of Advanced Scientific Computing Research, Quantum Algorithms Teams and Quantum Testbed Pathfinder programs.
\end{acknowledgments}


%


\clearpage
\appendix
\onecolumngrid
\begin{center}
\vspace*{.5\baselineskip}
{\textbf{\large Supplemental Material: \\[3pt] Efficiently computable bounds for magic state distillation }}\\[1pt] \quad \\
\end{center}

This supplemental material provides more detailed analysis and proofs of the results stated in the main text. On occasion we reiterate some of the steps in the main text in order to make the supplemental material more explicit and self-contained.

\section{The stabilizer formalism and discrete Wigner function\label{app: stabilizer formalism}}

\label{app:stabilizer-background}

For most known fault-tolerant schemes, the restricted set of quantum operations is the stabilizer operations, consisting of preparation and measurement in the computational basis and a restricted set of unitary operations. Here we review the basic elements of the stabilizer formalism for quantum systems of prime dimension.

We denote a Hilbert space of dimension $d$ by $\cH_d$ and the standard computational basis by $\{\ket j\}_{j=0,\ldots,d-1}$. For a prime number $d$, we define the respective shift and boost  operators $X,Z\in\cL(\cH_d)$ as
\begin{align}
X\ket j &= \ket{j\oplus 1} , \\
Z\ket j &=\omega^j \ket j, \qquad \omega=e^{2\pi i /d},
\end{align}
where $\oplus$ denotes addition modulo $d$.
We define the Heisenberg--Weyl operators as
\begin{align}
T_{\bu}= \tau^{-a_1a_2}Z^{a_1}X^{a_2},
\end{align}
where $\tau=e^{(d+1)\pi i/d}$ and $\bu=(a_1,a_2)\in \ZZ_d\times \ZZ_d$.

For a system with composite Hilbert space $\cH_{A}\ox\cH_B$, the Heisenberg--Weyl operators are
the tensor product of the subsystem Heisenberg--Weyl operators:
\begin{align}
T_{\bu_A\oplus \bu_B} = T_{\bu_A} \ox T_{\bu_B},
\end{align} 
where $\bu_A\oplus \bu_B$ is an element of $\ZZ_{d_A}\times\ZZ_{d_A}\times\ZZ_{d_B}\times\ZZ_{d_B}$.

The set $\cC_d$ of Clifford operators  is defined to be the set of unitary operators that map Heisenberg--Weyl operators to Heisenberg--Weyl operators under conjugation up to phases:
\begin{align}
U\in \cC_d \text{ iff } \forall \bu, \exists \, \theta,\bu', \text{ such that }
UT_\bu U^{\dagger} = e^{i\theta}T_{\bu'}.
\end{align}
The Clifford operators form the Clifford group. The Clifford group for qutrits can be generated by the qutrit Hadamard gate in \eqref{eq:qutrit-hadamard}, the phase gate $S = \proj{0} + \proj{1} + \omega \proj{2}$, and the sum gate $\sum_{i,j=0}^{2} \ket{i}\bra{i} \otimes \ket{i\oplus j}\bra{j}$ \cite{PhysRevA.71.042315}.

The pure stabilizer states can be obtained by applying Clifford operators to the state $\ket 0$:
\begin{align}
\{S_j\}_j=
\{U\proj0U^\dagger: U\in \cC_d \}.
\end{align}
The set of stabilizer states is the convex hull of the set of pure stabilizer states:
\begin{align}
\STAB(\cH_d)=\left\{\rho\in\cS(\cH_d): \rho=\sum_j p_j S_j, \ \forall j \ p_j\ge 0, \ \sum_j p_j = 1\right\}.
\end{align}

The discrete Wigner function \cite{Wootters1987,Gross2006,Gross2007} was used to show the existence of bound magic states \cite{Veitch2012}, which are states from which it is impossible to distill magic at a strictly positive rate. 
For an overview of discrete Wigner functions, we refer to \cite{Veitch2012,Veitch2014,Ferrie2011} for more details.
See also \cite{Ferrie2011} for quasi-probability representations of quantum
theory with applications to quantum information
science. 

For each point $\bu$ in the discrete phase space, there is a corresponding operator $A_\bu$, and the value of the discrete Wigner representation of $\rho$ at this point is given by
\begin{align}
W_{\rho}(\bu)=\frac1d \tr A_\bu\rho,
\end{align}
where $\{A_\bu\}_\bu$ are the phase-space point operators:
\begin{align}
A_\bu :=T_\bu A_0 T_\bu^\dag, \quad \quad A_0 :=\frac1d \sum_\bu T_\bu
.
\end{align}
In general, we define the discrete Wigner function of an operator $V$ as $W_V(\bu):=\frac1d \tr A_\bu V$.
Some properties of the set $\{A_\bu\}_\bu$ of phase-space point operators are listed as follows:
\begin{enumerate}
\item  $A_\bu$ is Hermitian;
\item $\sum_\bu A_\bu/d=\1$;
\item $\tr A_\bu A_{\bu'}= d\ \delta (\bu,\bu')$ ;
\item $\tr A_\bu=1$;
\item $V=\sum_\bu W_{V}(\bu) A_\bu$;
\item $\tr V Y = d \sum_{\bu} W_{V}(\bu) W_{Y}(\bu)$, for $d$-dimensional operators $V$ and $Y$.
\end{enumerate}
 
\begin{definition}
The Wigner trace norm of a Hermitian operator $V$ is defined as
\begin{align}
\|V\|_{W,1}:= \sum_\bu |W_V(\bu)| = \sum_\bu |\tr A_\bu V/d|.
\end{align}
\end{definition}

\begin{definition}
The Wigner spectral norm of a Hermitian operator $V$ acting on a space of dimension $d$ is defined as
\begin{align}\label{eq: inf norm wigner}
\|V\|_{W,\infty}:=d \max_\bu |W_V(\bu)| = \max_\bu |\tr A_\bu V|.
\end{align}
\end{definition}

From these definitions and duality of norms, we find that
\begin{align}
\|V\|_{W,1} & = \max_C \{ | \tr VC | : \|C\|_{W,\infty} \leq 1\} ,\\
\|V\|_{W,\infty} & = \max_C \{ | \tr VC | : \|C\|_{W,1} \leq 1\} .
\end{align}

\section{Efficient computation of thauma via convex optimization}

\label{app:eff-comp-thauma}

Let us recall that the thauma of a quantum state $\rho$ is given by 
 \begin{equation}
 \theta (\rho) := \min_{\sigma \in \cW} D(\rho\|\sigma) .
 \end{equation}
Note that $\theta(\rho)$ is a convex optimization problem with relative entropy as the objective function, due to the fact that the relative entropy is convex in $\sigma$ and the constraint $\sigma \in \cW$ is convex. Indeed, the minimization over $\sigma \in \cW$ can be equivalently written as a minimization over the following set:
\begin{equation}
\{\sigma, \{c_\bu\}_{\bu\in \ZZ_d\times \ZZ_d}  : \forall \bu, -c_\bu\le W_{\sigma}(\bu) \le c_\bu, \ \sum_{\bu} c_\bu \le 1, \ \sigma \ge 0\}.
\end{equation}
Then using the efficient optimization algorithm of the quantum relative entropy from \cite{Fawzi2017a}, we provide,  in our ancillary files on the arXiv, a program to compute the thauma $\theta(\rho)$ of a quantum state $\rho$.

\section{Generalized thauma measures}

\label{app:gen-thauma}

In the main text, we discussed the notion of generalized thauma as arising from any relative-entropy like measure that is non-increasing under the action of a quantum channel. Furthermore, we mentioned how one could have Petz--R\`enyi and sandwiched R\'enyi thauma measures. Here, we explicitly give the definitions of such measures.

Recall that the Petz--R\'enyi relative entropy of a state $\rho$ and a positive semi-definite operator $\sigma$ is defined for $\alpha \in (0,1)\cup(1,\infty)$ as \cite{P86}
\begin{align}
D_\alpha(\rho\|\sigma)\coloneqq\frac{1}{\alpha-1}\log_2\tr\left[\rho^\alpha\sigma^{1-\alpha}\right],
\end{align}
whenever either $\alpha\in(0,1)$ and $\rho$ is not orthogonal to $\sigma$ in Hilbert-Schmidt inner product or $ \alpha>1$ and $\supp(\rho)\subseteq\supp(\sigma)$. Otherwise we set $D_\alpha(\rho\|\sigma)\coloneqq\infty$.
In the above and throughout, we employ the convention that inverses are to be understood as generalized inverses.
For $\alpha\in\{0,1\}$, the Petz-R\'enyi relative entropy is defined in the limit as
\begin{align}
D_0(\rho\|\sigma)& \coloneqq\lim_{\alpha\to0}D_\alpha(\rho\|\sigma)=-\log_2
\tr\left[ P_{\rho} \sigma\right],
\\
D_1(\rho\|\sigma)& \coloneqq\lim_{\alpha\to1}D_\alpha(\rho\|\sigma)=D(\rho\|\sigma),
\label{eq:Petz-to-Umegaki}
\end{align}
where $P_{\rho}$ denotes the projection onto the support of $\rho$. The Petz--R\'enyi relative entropies are ordered, in the sense that
\begin{equation}
D_\alpha(\rho\|\sigma) \geq D_\beta(\rho\|\sigma),
\label{eq:Petz-Renyi-ordered}
\end{equation}
for $\alpha \geq \beta \geq 0$ \cite{Tomamichel2009,Mosonyi2011}.

The sandwiched R\'enyi relative entropy is defined for a state $\rho$ and a positive semi-definite operator $\sigma$ and $\alpha\in(0,1)\cup(1,\infty)$  as~\cite{Muller-Lennert2013,Wilde2014a}
\begin{align}
\widetilde{D}_\alpha(\rho\|\sigma)\coloneqq\frac{1}{\alpha-1}\log_2\tr\left[\left(\sigma^{\frac{1-\alpha}{2\alpha}}\rho\sigma^{\frac{1-\alpha}{2\alpha}}\right)^\alpha\right]
\end{align}
whenever either $ \alpha\in(0,1)$ and $\rho$ is not orthogonal to $\sigma$ in Hilbert-Schmidt inner product or $\alpha>1$ and $\supp(\rho)\subseteq\supp(\sigma)$. Otherwise we set $\widetilde{D}_\alpha(\rho\|\sigma)\coloneqq\infty$. For $\alpha=1$, we define the sandwiched R\'enyi relative entropy in the limit as \cite{Muller-Lennert2013,Wilde2014a}
\begin{align}
\widetilde{D}_1(\rho\|\sigma)\coloneqq\lim_{\alpha\to1}\widetilde{D}_\alpha(\rho\|\sigma)=D(\rho\|\sigma).
\label{eq:sandwiched-to-Umegaki}
\end{align}
We have that
\begin{align}
\widetilde{D}_{1/2}(\rho\|\sigma)=-\log_2 F(\rho,\sigma),
\end{align}
with Uhlmann's fidelity defined as $F(\rho,\sigma)\coloneqq\|\sqrt{\rho}\sqrt{\sigma}\|_1^2$ \cite{Uhlmann1976}.
In the limit $\alpha\to\infty$, the sandwiched R\'enyi relative entropy converges to the max-relative entropy~\cite{Datta2009}
\begin{align}
D_{\max}(\rho\|\sigma) \coloneqq\widetilde{D}_\infty(\rho\|\sigma)& \coloneqq\lim_{\alpha\to\infty}\widetilde{D}_\alpha(\rho\|\sigma)
\\
& =\log_2 \left\|\sigma^{-1/2}\rho\sigma^{-1/2}\right\|_\infty \\
& =\inf\left\{\lambda:\rho\leq2^{\lambda}\cdot\sigma\right\},
\end{align}
as shown in \cite{Muller-Lennert2013}. The  sandwiched R\'enyi relative entropies are ordered, in the sense that
\begin{equation}
\widetilde{D}_\alpha(\rho\|\sigma) \geq \widetilde{D}_\beta(\rho\|\sigma),
\label{eq:sandwiched-Renyi-ordered}
\end{equation}
for $\alpha \geq \beta \geq 0$ \cite{Muller-Lennert2013}.

The Petz--R\'enyi relative entropy obeys the data-processing inequality for all $\alpha \in [0,1)\cup(1,2]$ \cite{P86}:
\begin{equation}
D_\alpha(\rho\|\sigma) \geq D_\alpha(\cN(\rho)\|\cN(\sigma)),
\end{equation}
 it is non-negative for a state $\rho$ and a subnormalized state $\sigma$, and equal to zero if $\rho = \sigma$. Thus, the Petz--R\'enyi thauma defined as
 \begin{equation}
 \theta_\alpha (\rho) := \min_{\sigma \in \cW} D_\alpha(\rho\|\sigma)
 \end{equation}
 is a magic monotone. We remark here that the min-thauma $\theta_{\min}(\rho)$ employed in the main text is equal to $\theta_0 (\rho)$.
 
The sandwiched R\'enyi relative entropy obeys the data-processing inequality for all $\alpha \in [1/2,1)\cup(1,\infty]$ \cite{FL13,Beigi2013}:
\begin{equation}
\widetilde{D}_\alpha(\rho\|\sigma) \geq \widetilde{D}_\alpha(\cN(\rho)\|\cN(\sigma)),
\end{equation}
 it is non-negative for a state $\rho$ and a subnormalized state $\sigma$, and equal to zero if $\rho = \sigma$. (See \cite{Wilde2018a} for an alternative proof of the data processing inequality.) Thus, the sandwiched R\'enyi thauma defined as
 \begin{equation}
 \widetilde{\theta}_\alpha (\rho) := \min_{\sigma \in \cW} \widetilde{D}_\alpha(\rho\|\sigma)
 \end{equation}
 is a magic monotone. We remark here that the max-thauma $\theta_{\max}(\rho)$ employed in the main text is equal to $\lim_{\alpha \to \infty}\widetilde{\theta}_\alpha (\rho)$.
 
From the ordering of the R\'enyi relative entropies in \eqref{eq:Petz-Renyi-ordered} and \eqref{eq:sandwiched-Renyi-ordered} and the limits in \eqref{eq:Petz-to-Umegaki} and \eqref{eq:sandwiched-to-Umegaki}, we conclude the following ordering of some of the thauma measures:
\begin{equation}
\theta_{\min}(\rho) \leq \theta(\rho) \leq \theta_{\max}(\rho).
\label{eq:ordering-thauma}
\end{equation}


 \section{SDPs of min- and max-thauma and  proof of Proposition~\ref{prop: min theta add}}
 
 \label{app:min-thauma-proofs}
 
  \begingroup
\renewcommand{\theproposition}{\ref{prop: min theta add}
}
We first derive the dual SDP for the min-thauma:
\begin{align}
	\theta_{\min}(\rho)&= -\log_2 \max\{\tr P_\rho\sigma: \sum_{\bu} |W_{\sigma}(\bu)|\le 1,\  \sigma \ge 0 \} \label{eq:theta max 1}\\
	&=-\log_2 \max\{\tr P_\rho\sigma:\forall \bu, -c_\bu\le W_{\sigma}(\bu) \le c_\bu, \ \sum_{\bu} c_\bu \le 1, \ \sigma \ge 0 \} \label{eq:min dual 1}\\
& =-\log_2 \min\{t:\forall \bu, \ x_\bu+y_\bu\le t, \ x_\bu,y_\bu\ge 0, \ \sum_{\bu} \frac{1}{d}A_\bu(x_\bu-y_\bu)\ge P_\rho \} \label{eq:min dual 2}\\
&=-\log_2 \min\{t:\forall \bu, \ |q_\bu|\le t, \ \sum_{\bu} \frac{1}{d}A_\bu q_\bu\ge P_\rho \} \label{eq:min dual 3}\\
&=-\log_2 \min\{t:\forall \bu,  \ -t  \le \tr A_{\bu}Q \le t,\  Q\ge P_\rho \} \label{eq:min dual 4}
\end{align}
In the above proof, Eq.~\eqref{eq:min dual 1} follows from the definition of the Wigner trace norm. Eq.~\eqref{eq:min dual 2} utilizes the duality of SDPs. Eq.~\eqref{eq:min dual 3} introduces a new variable $q_\bu=x_{\bu}-y_\bu$.

Thus, adapting the definition of the Wigner spectral norm in \eqref{eq: inf norm wigner}, we find that 
\begin{align}
\theta_{\min}(\rho)&=-\log_2 \min\{t:\forall \bu,  -t  \le \tr A_{\bu}Q \le t,\ Q\ge P_\rho \}\\
&= -\log_2 \min\{ \|Q\|_{W,\infty}: Q\ge P_\rho\} .
 \label{eq:Rmin min}
\end{align}

  For max-thauma, its primal and dual SDPs are as follows:
\begin{align}
	\theta_{\max}(\rho)&= \log_2 \max\{\tr \rho G: \|G\|_{W,\infty}\le 1,\ G\ge 0\} \label{eq:rmax max}\\
	&= \log_2 \min\{\tr V: \rho \le V, \ \|V\|_{W,1}\le 1\}. \label{eq:rmax min}
	\end{align}

 \begin{proposition}[Additivity of min- and max-thauma]
 For states $\rho$ and $\tau$, the following additivity equality holds
\begin{align}
\theta_{\min}(\rho\ox\tau) & =\theta_{\min}(\rho)+\theta_{\min}(\tau),\\
\theta_{\max}(\rho\ox\tau) & =\theta_{\max}(\rho)+\theta_{\max}(\tau).
\end{align}
Consequently,
$\theta_{\min}(\rho) \leq \theta^{\infty}(\rho) \leq \theta_{\max}(\rho) $.
\end{proposition}
\begin{proof}
  Suppose that the optimal solutions to the SDP~\eqref{eq:theta max 1} of states $\rho_1$ and $\rho_2$ are $\sigma_1$ and $\sigma_2$, respectively. It is clear that 
  \begin{align}
  	\|\sigma_1\ox \sigma_2\|_{W,1}= \|\sigma_1\|_{W,1}\times \|\sigma_2\|_{W,1}\le 1,
  \end{align}
  which means that
  \begin{align}
  \theta_{\min}(\rho_1\ox \rho_2)&\le -\log_2 \tr (P_{\rho_1}\ox P_{\rho_2})(\sigma_1\ox \sigma_2)\\ &=\theta_{\min}(\rho_1)+\theta_{\min}(\rho_2).
  \end{align}
  
  On the other hand, suppose that the optimal solutions to the SDP~\eqref{eq:Rmin min} of $\rho_1$ and $\rho_2$ are $Q_1$ and $Q_2$, respectively. It is clear that 
  \begin{align}
 	Q_1 \ox Q_2 \ge P_1\ox P_2,
  \end{align}
  which means that
  \begin{align}
  \theta_{\min}(\rho_1\ox \rho_2)& \ge -\log_2 \|Q_1\ox Q_2\|_{W,\infty}= -\log_2(\|Q_1\|_{W,\infty}\times \|Q_2\|_{W,\infty})\\ &=\theta_{\min}(\rho_1)+\theta_{\min}(\rho_2).
  \end{align}
  
The proof of additivity of max-thauma is similar to the proof of additivity given above for min-thauma. The idea is to use both the primal and dual SDPs of the max-thauma.

The ordering property $\theta_{\min}(\rho) \leq \theta^{\infty}(\rho) \leq \theta_{\max}(\rho) $ is a consequence of the ordering in  \eqref{eq:ordering-thauma} and additivity of min- and max-thauma.
\end{proof}
\endgroup

\section{Monotonicity of max-thauma on average under stabilizer operations}
\label{app:max-thauma-props}

\begin{proposition}
\label{prop:max-thauma-additive-monotone}
The max-thauma does not increase on average under stabilizer operations. That is, suppose that there is a general stabilizer protocol $\cN=\{ \cN_j\}_j$ that transforms $\rho$ to $\sigma_j=\cN_j(\rho)/p_j$ with probability $p_j=\tr \cN_j(\rho)$, where $\cN_j$ is a completely positive stabilizer map and $\sum_j \cN_j$ is trace preserving. Then
\begin{equation}
\theta_{\max}(\rho) \geq \sum_j  p_j {\theta_{\max}(\sigma_j)}.
\end{equation}
\end{proposition}
\begin{proof}
Let us suppose $V$ is the optimal solution to the SDP~\eqref{eq:rmax min}. From complete positivity of the map $\cN_j$, it follows that
\begin{align}
\cN_j (V) /p_j \ge \cN_j(\rho) / p_j,
\end{align}
which means that $\cN_j (V) / p_j$ is feasible for the  SDP~\eqref{eq:rmax min} of 
$\theta_{\max}(\sigma_j)$. Thus,
 \begin{align}
 2^{\theta_{\max}(\sigma_j)} \le \|\cN_j (V) / p_j\|_{W,1},
 \end{align}
 and we find that
 \begin{align}
 \theta_{\max}(\rho) &= \log_2 \|V\|_{W,1}\\
 & \ge \log_2 \sum_j \|\cN_j(V)\|_{W,1}\\
  & = \log_2 \sum_j p_j \|\cN_j(V)/p_j\|_{W,1}\\
 &\ge \log_2 \sum_j p_j 2^{\theta_{\max}(\sigma_j)}\\
 & \ge\sum_j  p_j \log_2 2^{\theta_{\max}(\sigma_j)}\\
 &=\sum_j  p_j {\theta_{\max}(\sigma_j)}.
 \end{align}
 The first inequality follows from an inequality stated in the proof of Theorem 4 of~\cite{Veitch2014}. It can be understood as being a consequence of the fact that the stabilizer operation $\{\cN_j\}_j$ has a discrete Wigner representation as a conditional probability distribution, which does not increase the classical one norm. 
\end{proof}


\section{Min-, regularized, and max-thaumas of Strange state, Norrell state, $H_+$ state, and $T$ state, and proof of Theorem~\ref{th: S inequal N}}

\label{app:example-states-irrev}

\begingroup
\renewcommand{\theproposition}{\ref{prop: thauma collapse magic states}}

Let us first recall the definitions of the Strange state, Norrell state, $H_+$ state, and $T$ state:
\begin{align}
\ket \SSS & = \frac{1}{\sqrt 2} (\ket 1-\ket 2), \\ 
\ket \NN & = \frac{1}{\sqrt 6} (-\ket 0+2\ket 1-\ket 2),\\
\ket{H_+} & = \frac{(1+\sqrt{3})\ket 0 + \ket 1 + \ket{2}}{\sqrt{2(3+\sqrt{3})} },\\
\ket T & = \frac{1}{\sqrt 3}(\xi  \ket  0+ \ket 1 + \xi^{-1}\ket 2),
\end{align}
where $\xi = e^{2 \pi i / 9}$ and we recall that $H_+$ magic state is the $+1$-eigenstate of the $H$ gate defined in \eqref{eq:qutrit-hadamard}.

\begin{proposition}
For the Strange state, Norrell state, $H_+$ state, and $T$ state, their min-thauma, regularized thauma, max-thauma, and in some cases, regularized relative entropy of magic can be evaluated exactly as
\begin{align}
  \theta_{\min}(\SSS)  & = \theta^{\infty}(\SSS) = \theta_{\max}(\SSS) = \log_2(5/3), \label{eq:thauma-Strange-state}\\
R_{\cM}^{\infty}(\NN) = \theta_{\min}(\NN) & = \theta^{\infty}(\NN) = \theta_{\max}(\NN)  =  \log_2(3/2) ,
\label{eq:thauma-Norrell-state} \\
R_{\cM}^{\infty}(H_+) = \theta_{\min}(H_+) &  = \theta^{\infty}(H_+) = \theta_{\max}(H_+) 
  = \log_2(3-\sqrt 3), \label{eq:thauma-H-state} \\
 \theta_{\min}(T) & = \theta^{\infty}(T) = \theta_{\max}(T) = \log_2(1+2\sin(\pi/18)).
\label{eq:thauma-T-state}
\end{align}
\end{proposition}

\begin{proof}
Our proof strategy for all four cases is the same. For each case, we first bound the min-thauma of the given state from below, and then we bound its max-thauma from above. In each case, we find that the lower and upper bounds match, allowing us to conclude the equalities. If the optimal state for the max-thauma is a stabilizer state, then one can conclude from definitions that max-thauma is an upper bound on the regularized relative entropy of magic, and so this is how we obtain equalities involving regularized relative entropy of magic in some cases. We remark here that we found these results by first numerically solving the SDPs for min- and max-thauma, and then determining an analytic form for the optimal values and solutions for the SDPs. We have made our Matlab and Mathematica source files available (see ancillary files on arXiv). We numerically solved the SDPs using Matlab and then we confirmed the analytic solutions using Mathematica.

\bigskip
\textbf{Strange state.} We begin with the Strange state $\SSS$, by bounding its min-thauma $\theta_{\min}( {\SSS})$ from below.
First, we choose a feasible solution to the dual SDP~\eqref{eq:Rmin min} as
\begin{align}
Q=\proj{\SSS}+\frac15\proj{0}+\frac{1}{10}(\proj 1+\proj 2+\ketbra12+\ketbra21).
\end{align}
As $\|Q\|_{W,\infty}=\frac3 5$, it holds that
\begin{align}
\theta_{\min}( {\SSS})\ge - \log_2 \|Q\|_{W,\infty}=\log_2 \frac 5  3.
\end{align}

We now bound the max-thauma $\theta_{\max}( {\SSS})$ from above by the same value. Consider that the sub-normalized  state $\sigma=\frac{3}{5}\SSS$ satisfies $\|\sigma\|_{W,1}=1$. This implies that
\begin{align}
\theta_{\max}({\SSS}) \leq \log_2 \Vert \sigma^{-1/2} \SSS \sigma^{-1/2} \Vert_\infty = \log_2 \frac 5 3.
\end{align}

By exploiting the above bounds and the inequalities $\theta_{\min}({\SSS})\le \theta^{\infty}({\SSS}) \leq \theta_{\max}({\SSS})$
from Proposition~\ref{prop: min theta add}, we conclude the equalities in \eqref{eq:thauma-Strange-state}.

\bigskip
\textbf{Norrell state.} We turn to the Norrell state $\NN$, and we bound its min-thauma from below. Let us take
\begin{equation}
Q=\proj {\NN}+\frac{1}{6}(\proj {0} - \ketbra 0 2 -\ketbra 2 0 + \proj {2})
\end{equation}
as a feasible solution to the SDP~\eqref{eq:Rmin min}. By direct calculation, we find that $\|Q\|_{W,\infty} = 2/3$, so that
\begin{align}\label{eq:N state lower}
\theta_{\min} ( {\NN}) \ge -\log_2 \|Q\|_{W,\infty}= \log_2 \frac{3}{2}.
\end{align}

We now bound the max-thauma of the Norrell state from above. Let us choose the following stabilizer state:
\begin{align}
\tau = \frac{1}{3} (\proj {u_0}+\proj {u_1}+\proj {u_2}) =
\frac{1}{9}\left(\begin{matrix}
2 & -1 & 2\\
-1 & 5 & -1 \\
2 & -1 & 2
\end{matrix}\right)
,
\end{align}
which is a convex mixture of the following stabilizer states:
\begin{align}
	\ket {u_0} & = \ket 1,\\
	\ket {u_1} & =  (\ket 0 +e^{2 \pi i/3}\ket 1 +\ket 2)/\sqrt{3},\\
	\ket {u_2} & =  (\ket 0 +e^{4 \pi i/3}\ket 1 +\ket 2)/\sqrt{3}.
\end{align}
The state $\ket {u_0} = X \ket 0$, $\ket {u_1} = SHX \ket 0$, and $\ket{u_2} = S^2 H X^2 \ket 0$. Thus, each state is a stabilizer state as claimed.
Then,
\begin{align}\label{eq:N state upper}
	\theta_{\max}(\NN) \leq D_{\max}(\proj \NN \| \tau) = \log_2 \frac{3}{2}.
\end{align}
The last equality follows by direct calculation.

By exploiting \eqref{eq:N state lower} and \eqref{eq:N state upper} and the inequalities $\theta_{\min}({\NN})\le \theta^{\infty}({\NN}) \leq \theta_{\max}({\NN})$
from Proposition~\ref{prop: min theta add}, we conclude the equalities
$\theta_{\min}({\NN})= \theta^{\infty}({\NN}) = \theta_{\max}({\NN}) =  \log_2 \frac{3}{2}$
in \eqref{eq:thauma-Norrell-state}. The inequality $\log_2 \frac{3}{2} = \theta^{\infty}(\NN) \leq R_{\cM}^{\infty}(\NN) $ follows by definition. The opposite inequality follows because $\tau$ is a stabilizer state, so that for all $n\geq 1$,
\begin{equation}
\frac{1}{n} R_{\cM}(\NN^{\otimes n}) \leq \frac{1}{n}D_{\max}((\proj \NN)^{\otimes n} \| \tau^{\otimes n}) = D_{\max}(\proj \NN \| \tau) = \log_2 \frac{3}{2},
\end{equation}
which implies that $R_{\cM}^{\infty}(\NN) \leq \log_2 \frac{3}{2}$.

\bigskip
\textbf{$H_+$ state.} We consider the $H_+$ state, and we bound its min-thauma from below. Let us take
\begin{equation}
Q=\proj {H_{+}}+\frac{1}{3+\sqrt 3}\proj {H_{i}}
\end{equation}
as a feasible solution to the SDP~\eqref{eq:Rmin min}. By direct calculation, we find that 
\begin{align}
\|Q\|_{W,\infty}= \frac{1}{3-\sqrt 3}. 
\end{align}
Thus, we have that
\begin{align}
\label{eq: R H lower}
\theta_{\min}(H_+) \ge - \log_2 \|Q\|_{W,\infty} = \log_2 (3-\sqrt 3).
\end{align}

We now bound the max-thauma of the $H_+$ state from above. Let us choose the stabilizer state 
\begin{align}
\sigma= \frac{1}{2} (\proj {v_0} + \proj{v_1})
\end{align}
 with
\begin{align}
\ket {v_0}=\ket 0, \quad \ket {v_1} = \frac{1}{\sqrt 3} (\ket 0+\ket 1 +\ket 2).
\end{align}
Then it is clear that
\begin{align}
\label{eq: R H upper}
\theta_{\max}(H_+) \leq  D_{\max}(\proj {H_{+}} \| \sigma) = \log_2 \bra{H_+}  \sigma^{-1} \ket{H_+}= \log_2 (3-\sqrt 3),
\end{align}
which can be confirmed by direct calculation, using the fact that
\begin{equation}
\sigma^{-1} =
\left(\begin{matrix}
2 &-1& -1\\
-1 & 2 & 2\\
-1 & 2 & 2 
\end{matrix}
\right).
\end{equation}

By exploiting \eqref{eq: R H lower} and \eqref{eq: R H upper} and the inequalities $\theta_{\min}(H_+)\le \theta^{\infty}(H_+) \leq \theta_{\max}(H_+)$
from Proposition~\ref{prop: min theta add}, we conclude the equalities
$\theta_{\min}(H_+)= \theta^{\infty}(H_+) = \theta_{\max}(H_+) = \log_2 (3-\sqrt 3)$
in \eqref{eq:thauma-H-state}.
The inequality $\log_2 (3 -\sqrt{3}) = \theta^{\infty}(H_+) \leq R_{\cM}^{\infty}(H_+) $ follows by definition. The opposite inequality follows because $\sigma$ is a stabilizer state, so that for all $n\geq 1$,
\begin{equation}
\frac{1}{n} R_{\cM}(H_+^{\otimes n}) \leq \frac{1}{n}D_{\max}((H_+)^{\otimes n} \| \sigma^{\otimes n}) = D_{\max}(H_+ \| \sigma) = \log_2 (3 - \sqrt{3}),
\end{equation}
which implies that $R_{\cM}^{\infty}(H_+) \leq \log_2 (3-\sqrt{3})$.

\bigskip
\textbf{$T$ state.} We consider the $T$ state, and we bound its min-thauma from below. Let us take
\begin{equation}
Q = \proj{T} + 2\sin(\pi/18) \proj{\phi}
\end{equation}
where
\begin{equation}
\ket{\phi} = (e^{-2\pi i/9}\ket 0 + e^{8 \pi i / 9}\ket 1 + \ket 2)/\sqrt{3}.
\end{equation}
Then one finds that $(\Vert Q \Vert_{W,\infty})^{-1} = 1+2\sin(\pi/18)$, so that
\begin{align}
\theta_{\min}(T) \ge - \log_2 \|Q\|_{W,\infty} = \log_2 (1+2\sin(\pi/18)).
\label{eq: R T lower}
\end{align} 

We now bound the max-thauma of the $T$ state from above. Let us choose the following state $\tau$ with positive Wigner function:
\begin{align}
\tau=\left(\begin{matrix}
1/3 &   r_1 & r_2 \\
r_1^* & 1/3 & r_1\\
r_2^* & r_1^* & 1/3 
\end{matrix}
\right),
\end{align}
with
\begin{align}
r_1  =\frac{e^{\pi i / 9}}{6 \cos(2 \pi / 9)},\qquad \qquad
r_2  =\frac{e^{5 \pi i / 9}}{6 \cos(2 \pi / 9)}.
\end{align}
The state $\sigma$ has eigenvalues equal to $0$, $\frac{2}{\sqrt{3}}\sin(2 \pi / 9)\approx 0.742227$, and $\frac{2\sin(\pi/9)\sin(2 \pi /9)}{\sqrt{3}\cos(\pi/18)} \approx 0.257773 $, and it is thus positive semi-definite as claimed.
A calculation gives that the discrete Wigner function of the above state is described by the following matrix:
\begin{equation}
\left(\begin{matrix}
s & 1/3 -s & s\\
0 & s & 0 \\
1/3 - s & 0 & 1/3 -s
\end{matrix}\right)
\end{equation}
where
\begin{equation}
s = \frac{1}{9} \left( 2 +  \frac{\sin(\pi/18)}{ \cos(2 \pi /9)}\right) \approx 0.247409,
\end{equation}
confirming that it has only non-negative entries as claimed.
Therefore,
\begin{align}
\theta_{\max}(T) \le D_{\max}(\proj{T}\|\tau) \le \log_2(1+2\sin(\pi/18)).
\label{eq: R T upper}
\end{align}
The last inequality can be verified from the fact that the operator $(1+2\sin(\pi/18))\tau - \proj{T}$ is positive semi-definite, having two zero eigenvalues and one non-zero eigenvalue equal to $\frac{\sin(\pi/9)}{\cos(\pi/18)}\approx 0.347296$.

By exploiting \eqref{eq: R T lower} and \eqref{eq: R T upper} and the inequalities $\theta_{\min}(T)\le \theta^{\infty}(T) \leq \theta_{\max}(T)$
from Proposition~\ref{prop: min theta add}, we conclude the equalities in \eqref{eq:thauma-T-state}.
\end{proof}

\endgroup

\section{Proof of Proposition~\ref{th:hypothesis bound magic}}

 \begingroup
 \renewcommand{\theproposition}{\ref{th:hypothesis bound magic}}

\begin{proposition}
Given a state $\rho$, the following holds
\begin{align}
\cM_{H_+}^\ve(\rho) & \le \frac{\min_{\sigma\in \cW}D_H^\ve(\rho\|\sigma)}{\log_2(3-\sqrt3)},
\\
\cM_{T}^\ve(\rho) & \le \frac{\min_{\sigma\in \cW}D_H^\ve(\rho\|\sigma)}{\log_2(1+2\sin(\pi/18))}.
\end{align}
\end{proposition}

\begin{proof}
   Let us suppose that the stabilizer operation $\Lambda$ satisfies
$   \Lambda(\rho) \approx_\ve \proj {{H_+}}^{\ox k}$ for some $k$. 
Recall that for any   $\tau \in \cW $, we have that
$
    \tr [ \proj {\phi}^{\ox n} \tau ] \leq 2^{-n\theta_{\min}(\phi)}
$.
   Given $\sigma \in \cW$, and using the above fact, we have 
  \begin{align}
  D_0(\proj {{H_+}} ^{\ox k}\|\Lambda(\sigma)) 
  & = -\log_2 \tr \proj {{H_+}}^{\ox k}\Lambda(\sigma) \\
  & \ge -\log_2 (3-\sqrt 3)^{-k}=k\log_2(3-\sqrt 3).
  \end{align}
   Therefore, we have that
  \begin{align}
    k\log_2(3-\sqrt3) & \le D_0(\proj {{H_+}}^{\ox k} \|\Lambda(\sigma))  \\
    & \leq D_H^\ve(\Lambda(\rho)\|\Lambda(\sigma))\\
    & \leq D_H^\ve(\rho\|\sigma).
  \end{align}
  The second inequality follows from the definition of $D_H^\ve$. The third inequality follows from the data processing inequality of $D_H^\ve$.
  
Finally, minimizing over all $\sigma\in\cW$, we find that
  \begin{align}
  k\le  {\min_{\sigma\in \cW}D_H^\ve(\rho\|\sigma)}/{\log_2(3-\sqrt3)}.
  \end{align}
  Since $k$ is arbitrary, we arrive at the bound in \eqref{eq:H-magic-bnd}. The proof of \eqref{eq:T-magic-bnd} is similar.
  \end{proof}
\endgroup
 \section{Second-order converse}
 \label{app: second order}
 \begingroup
 \renewcommand{\theproposition}{\ref{coro:second order theorem}} 
 \begin{corollary} 
For a quantum state $\rho$, number $n$ of available source states, and error tolerance $\ve \in (0,1)$, the following inequality holds
\begin{align}
\cM_{H_+}^\ve(\rho^{\ox n})  \le \frac{n\theta(\rho)+\sqrt{nV(\rho\|\sigma)}\Phi^{-1}(\varepsilon)+ O(\log n)}{\log_2(3-\sqrt3)},
\end{align}
where $\sigma$ is a minimizer of $\theta(\rho)$, $\Phi^{-1}$ is the inverse of the cumulative normal distribution function, and $V(\rho\|\sigma)= \tr \rho (\log_2 \rho - \log_2 \sigma)^2 - D(\rho\|\sigma)^2$.
Similarly,
\begin{align}
\cM_{T}^\ve(\rho^{\ox n}) & \le \frac{n\theta(\rho)+\sqrt{nV(\rho\|\sigma)}\Phi^{-1}(\varepsilon)+ O(\log n)}{\log_2(1+2\sin(\pi/18))}.
\end{align}
\end{corollary}
\begin{proof}
Suppose $\sigma$ is any minimizer of the thauma $\theta(\rho)$. For given $n$ and $\ve$, 
     \begin{align}
        \cM_{H_+}(\rho^{\ox n}) &\le D_H^\ve(\rho^{\ox n}\big\| \sigma^{\ox n}) \notag\\
      & = n D(\rho\big\| \sigma) + \sqrt{nV(\rho\big\| \sigma)} \ \Phi^{-1}(\ve) + O(\log n).
    \end{align}
    The first line follows from the upper bound in Proposition~\ref{th:hypothesis bound magic}, while the second line uses the second-order expansion of quantum hypothesis testing relative entropy \cite{Tomamichel2013a,Li2014a}. 
 \end{proof}
 \endgroup

\section{Bound on efficiency of magic-state distillation}

\begingroup
\renewcommand{\theproposition}{\ref{th:eff bound}}
\begin{proposition}
The efficiency of distilling a magic state $\sigma$ from independent copies of the resource state $\rho$ is bounded from below by
\begin{align}
N_{\theta_{\max}}(\rho,\sigma):=\theta_{\max}(\sigma)/\theta_{\max}(\rho).
\end{align}
\end{proposition}
\begin{proof}
The proof is in the same spirit of the mana bound in~\cite{Veitch2014}. Let us suppose that the stabilizer protocol $\Lambda$ outputs $\sigma$ from $\rho^{\ox n}$ with probability $p$. Without loss of generality, we can assume that
\begin{align}
\Lambda(\rho^{\ox n})=p\proj 0\ox \sigma+(1-p)\proj 1 \ox \tau,
\end{align}
where $\tau$ is some state.
Applying Proposition~\ref{prop:max-thauma-additive-monotone} and the fact that $\ket 0$ is a stabilizer state, we find that
\begin{align}
n\theta_{\max}(\rho)&=\theta_{\max}(\rho^{\ox n}) \\
                  &\ge p\theta_{\max}(\proj 0\ox \sigma)\\
                  & = p\theta_{\max}(\sigma),
\end{align}
which means that the following bound holds for all integer $n \geq 1$ and $p\in(0,1]$:
\begin{align}
\frac{n}{p}\ge  \frac{\theta_{\max}(\sigma)}{\theta_{\max}(\rho)}.
\end{align}
This concludes the proof.
\end{proof}
\endgroup


\begin{thebibliography}{85}%
\makeatletter
\providecommand \@ifxundefined [1]{%
 \@ifx{#1\undefined}
}%
\providecommand \@ifnum [1]{%
 \ifnum #1\expandafter \@firstoftwo
 \else \expandafter \@secondoftwo
 \fi
}%
\providecommand \@ifx [1]{%
 \ifx #1\expandafter \@firstoftwo
 \else \expandafter \@secondoftwo
 \fi
}%
\providecommand \natexlab [1]{#1}%
\providecommand \enquote  [1]{``#1''}%
\providecommand \bibnamefont  [1]{#1}%
\providecommand \bibfnamefont [1]{#1}%
\providecommand \citenamefont [1]{#1}%
\providecommand \href@noop [0]{\@secondoftwo}%
\providecommand \href [0]{\begingroup \@sanitize@url \@href}%
\providecommand \@href[1]{\@@startlink{#1}\@@href}%
\providecommand \@@href[1]{\endgroup#1\@@endlink}%
\providecommand \@sanitize@url [0]{\catcode `\\12\catcode `\$12\catcode
  `\&12\catcode `\#12\catcode `\^12\catcode `\_12\catcode `\%12\relax}%
\providecommand \@@startlink[1]{}%
\providecommand \@@endlink[0]{}%
\providecommand \url  [0]{\begingroup\@sanitize@url \@url }%
\providecommand \@url [1]{\endgroup\@href {#1}{\urlprefix }}%
\providecommand \urlprefix  [0]{URL }%
\providecommand \Eprint [0]{\href }%
\providecommand \doibase [0]{http://dx.doi.org/}%
\providecommand \selectlanguage [0]{\@gobble}%
\providecommand \bibinfo  [0]{\@secondoftwo}%
\providecommand \bibfield  [0]{\@secondoftwo}%
\providecommand \translation [1]{[#1]}%
\providecommand \BibitemOpen [0]{}%
\providecommand \bibitemStop [0]{}%
\providecommand \bibitemNoStop [0]{.\EOS\space}%
\providecommand \EOS [0]{\spacefactor3000\relax}%
\providecommand \BibitemShut  [1]{\csname bibitem#1\endcsname}%
\let\auto@bib@innerbib\@empty
\bibitem [{\citenamefont {Shor}(1997)}]{Shor1997}%
  \BibitemOpen
  \bibfield  {author} {\bibinfo {author} {\bibfnamefont {P.~W.}\ \bibnamefont
  {Shor}},\ }\href {\doibase 10.1137/S0097539795293172} {\bibfield  {journal}
  {\bibinfo  {journal} {SIAM Journal on Computing}\ }\textbf {\bibinfo {volume}
  {26}},\ \bibinfo {pages} {1484} (\bibinfo {year} {1997})},\ \Eprint
  {http://arxiv.org/abs/quant-ph/9508027} {arXiv:quant-ph/9508027 [quant-ph]}
  \BibitemShut {NoStop}%
\bibitem [{\citenamefont {Grover}(1996)}]{Grover1996}%
  \BibitemOpen
  \bibfield  {author} {\bibinfo {author} {\bibfnamefont {L.~K.}\ \bibnamefont
  {Grover}},\ }in\ \href {\doibase 10.1145/237814.237866} {\emph {\bibinfo
  {booktitle} {Proceedings of the twenty-eighth annual ACM symposium on Theory
  of computing - STOC '96}}}\ (\bibinfo  {publisher} {ACM Press},\ \bibinfo
  {address} {New York, New York, USA},\ \bibinfo {year} {1996})\ pp.\ \bibinfo
  {pages} {212--219},\ \Eprint {http://arxiv.org/abs/quant-ph/9605043}
  {arXiv:quant-ph/9605043 [quant-ph]} \BibitemShut {NoStop}%
\bibitem [{\citenamefont {Childs}\ and\ \citenamefont {van
  Dam}(2010)}]{Childs2010}%
  \BibitemOpen
  \bibfield  {author} {\bibinfo {author} {\bibfnamefont {A.~M.}\ \bibnamefont
  {Childs}}\ and\ \bibinfo {author} {\bibfnamefont {W.}~\bibnamefont {van
  Dam}},\ }\href {\doibase 10.1103/RevModPhys.82.1} {\bibfield  {journal}
  {\bibinfo  {journal} {Reviews of Modern Physics}\ }\textbf {\bibinfo {volume}
  {82}},\ \bibinfo {pages} {1} (\bibinfo {year} {2010})}\BibitemShut {NoStop}%
\bibitem [{\citenamefont {Lloyd}(1996)}]{Lloyd1996}%
  \BibitemOpen
  \bibfield  {author} {\bibinfo {author} {\bibfnamefont {S.}~\bibnamefont
  {Lloyd}},\ }\href {\doibase 10.1126/science.273.5278.1073} {\bibfield
  {journal} {\bibinfo  {journal} {Science}\ }\textbf {\bibinfo {volume}
  {273}},\ \bibinfo {pages} {1073} (\bibinfo {year} {1996})}\BibitemShut
  {NoStop}%
\bibitem [{\citenamefont {Childs}\ \emph {et~al.}(2018)\citenamefont {Childs},
  \citenamefont {Maslov}, \citenamefont {Nam}, \citenamefont {Ross},\ and\
  \citenamefont {Su}}]{Childs2018a}%
  \BibitemOpen
  \bibfield  {author} {\bibinfo {author} {\bibfnamefont {A.~M.}\ \bibnamefont
  {Childs}}, \bibinfo {author} {\bibfnamefont {D.}~\bibnamefont {Maslov}},
  \bibinfo {author} {\bibfnamefont {Y.}~\bibnamefont {Nam}}, \bibinfo {author}
  {\bibfnamefont {N.~J.}\ \bibnamefont {Ross}}, \ and\ \bibinfo {author}
  {\bibfnamefont {Y.}~\bibnamefont {Su}},\ }\href {\doibase
  10.1073/pnas.1801723115} {\bibfield  {journal} {\bibinfo  {journal}
  {Proceedings of the National Academy of Sciences}\ }\textbf {\bibinfo
  {volume} {115}},\ \bibinfo {pages} {9456} (\bibinfo {year} {2018})},\ \Eprint
  {http://arxiv.org/abs/1711.10980} {arXiv:1711.10980} \BibitemShut {NoStop}%
\bibitem [{\citenamefont {Shor}(1996)}]{Shor1996}%
  \BibitemOpen
  \bibfield  {author} {\bibinfo {author} {\bibfnamefont {P.~W.}\ \bibnamefont
  {Shor}},\ }in\ \href {\doibase 10.1109/SFCS.1996.548464} {\emph {\bibinfo
  {booktitle} {Proceedings of 37th Conference on Foundations of Computer
  Science}}}\ (\bibinfo  {publisher} {IEEE Comput. Soc. Press},\ \bibinfo
  {year} {1996})\ pp.\ \bibinfo {pages} {56--65},\ \Eprint
  {http://arxiv.org/abs/quant-ph/9605011} {arXiv:quant-ph/9605011 [quant-ph]}
  \BibitemShut {NoStop}%
\bibitem [{\citenamefont {Campbell}\ \emph {et~al.}(2017)\citenamefont
  {Campbell}, \citenamefont {Terhal},\ and\ \citenamefont
  {Vuillot}}]{Campbell2017c}%
  \BibitemOpen
  \bibfield  {author} {\bibinfo {author} {\bibfnamefont {E.~T.}\ \bibnamefont
  {Campbell}}, \bibinfo {author} {\bibfnamefont {B.~M.}\ \bibnamefont
  {Terhal}}, \ and\ \bibinfo {author} {\bibfnamefont {C.}~\bibnamefont
  {Vuillot}},\ }\href {\doibase 10.1038/nature23460} {\bibfield  {journal}
  {\bibinfo  {journal} {Nature}\ }\textbf {\bibinfo {volume} {549}},\ \bibinfo
  {pages} {172} (\bibinfo {year} {2017})}\BibitemShut {NoStop}%
\bibitem [{\citenamefont {Gottesman}(1997)}]{Gottesman1997}%
  \BibitemOpen
  \bibfield  {author} {\bibinfo {author} {\bibfnamefont {D.}~\bibnamefont
  {Gottesman}},\ }\href {http://arxiv.org/abs/quant-ph/9705052} {\bibfield
  {journal} {\bibinfo  {journal} {PhD thesis}\ } (\bibinfo {year} {1997})},\
  \Eprint {http://arxiv.org/abs/quant-ph/9705052} {arXiv:quant-ph/9705052
  [quant-ph]} \BibitemShut {NoStop}%
\bibitem [{\citenamefont {Aaronson}\ and\ \citenamefont
  {Gottesman}(2004)}]{Aaronson2004a}%
  \BibitemOpen
  \bibfield  {author} {\bibinfo {author} {\bibfnamefont {S.}~\bibnamefont
  {Aaronson}}\ and\ \bibinfo {author} {\bibfnamefont {D.}~\bibnamefont
  {Gottesman}},\ }\href {\doibase 10.1103/PhysRevA.70.052328} {\bibfield
  {journal} {\bibinfo  {journal} {Physical Review A}\ }\textbf {\bibinfo
  {volume} {70}},\ \bibinfo {pages} {052328} (\bibinfo {year} {2004})},\
  \Eprint {http://arxiv.org/abs/quant-ph/0406196} {arXiv:quant-ph/0406196
  [quant-ph]} \BibitemShut {NoStop}%
\bibitem [{\citenamefont {Gottesman}\ and\ \citenamefont
  {Chuang}(1999)}]{Gottesman1999}%
  \BibitemOpen
  \bibfield  {author} {\bibinfo {author} {\bibfnamefont {D.}~\bibnamefont
  {Gottesman}}\ and\ \bibinfo {author} {\bibfnamefont {I.~L.}\ \bibnamefont
  {Chuang}},\ }\href {\doibase 10.1038/46503} {\bibfield  {journal} {\bibinfo
  {journal} {Nature}\ }\textbf {\bibinfo {volume} {402}},\ \bibinfo {pages}
  {390} (\bibinfo {year} {1999})},\ \Eprint
  {http://arxiv.org/abs/quant-ph/9908010} {arXiv:quant-ph/9908010 [quant-ph]}
  \BibitemShut {NoStop}%
\bibitem [{\citenamefont {Zhou}\ \emph {et~al.}(2000)\citenamefont {Zhou},
  \citenamefont {Leung},\ and\ \citenamefont {Chuang}}]{Zhou2000}%
  \BibitemOpen
  \bibfield  {author} {\bibinfo {author} {\bibfnamefont {X.}~\bibnamefont
  {Zhou}}, \bibinfo {author} {\bibfnamefont {D.~W.}\ \bibnamefont {Leung}}, \
  and\ \bibinfo {author} {\bibfnamefont {I.~L.}\ \bibnamefont {Chuang}},\
  }\href {\doibase 10.1103/PhysRevA.62.052316} {\bibfield  {journal} {\bibinfo
  {journal} {Physical Review A}\ }\textbf {\bibinfo {volume} {62}},\ \bibinfo
  {pages} {052316} (\bibinfo {year} {2000})},\ \Eprint
  {http://arxiv.org/abs/quant-ph/0002039} {arXiv:quant-ph/0002039 [quant-ph]}
  \BibitemShut {NoStop}%
\bibitem [{\citenamefont {Bravyi}\ and\ \citenamefont
  {Kitaev}(2005)}]{Bravyi2005}%
  \BibitemOpen
  \bibfield  {author} {\bibinfo {author} {\bibfnamefont {S.}~\bibnamefont
  {Bravyi}}\ and\ \bibinfo {author} {\bibfnamefont {A.}~\bibnamefont
  {Kitaev}},\ }\href {\doibase 10.1103/PhysRevA.71.022316} {\bibfield
  {journal} {\bibinfo  {journal} {Physical Review A}\ }\textbf {\bibinfo
  {volume} {71}},\ \bibinfo {pages} {022316} (\bibinfo {year} {2005})},\
  \Eprint {http://arxiv.org/abs/quant-ph/0403025} {arXiv:quant-ph/0403025
  [quant-ph]} \BibitemShut {NoStop}%
\bibitem [{\citenamefont {Bravyi}\ and\ \citenamefont
  {Haah}(2012)}]{Bravyi2012}%
  \BibitemOpen
  \bibfield  {author} {\bibinfo {author} {\bibfnamefont {S.}~\bibnamefont
  {Bravyi}}\ and\ \bibinfo {author} {\bibfnamefont {J.}~\bibnamefont {Haah}},\
  }\href {\doibase 10.1103/PhysRevA.86.052329} {\bibfield  {journal} {\bibinfo
  {journal} {Physical Review A}\ }\textbf {\bibinfo {volume} {86}},\ \bibinfo
  {pages} {052329} (\bibinfo {year} {2012})},\ \Eprint
  {http://arxiv.org/abs/1209.2426} {arXiv:1209.2426 [quant-ph]} \BibitemShut
  {NoStop}%
\bibitem [{\citenamefont {Jones}(2013)}]{Jones2013}%
  \BibitemOpen
  \bibfield  {author} {\bibinfo {author} {\bibfnamefont {C.}~\bibnamefont
  {Jones}},\ }\href {\doibase 10.1103/PhysRevA.87.042305} {\bibfield  {journal}
  {\bibinfo  {journal} {Physical Review A}\ }\textbf {\bibinfo {volume} {87}},\
  \bibinfo {pages} {042305} (\bibinfo {year} {2013})},\ \Eprint
  {http://arxiv.org/abs/1303.3066} {arXiv:1303.3066} \BibitemShut {NoStop}%
\bibitem [{\citenamefont {Haah}\ \emph {et~al.}(2017)\citenamefont {Haah},
  \citenamefont {Hastings}, \citenamefont {Poulin},\ and\ \citenamefont
  {Wecker}}]{Haah2017}%
  \BibitemOpen
  \bibfield  {author} {\bibinfo {author} {\bibfnamefont {J.}~\bibnamefont
  {Haah}}, \bibinfo {author} {\bibfnamefont {M.~B.}\ \bibnamefont {Hastings}},
  \bibinfo {author} {\bibfnamefont {D.}~\bibnamefont {Poulin}}, \ and\ \bibinfo
  {author} {\bibfnamefont {D.}~\bibnamefont {Wecker}},\ }\href {\doibase
  10.22331/q-2017-10-03-31} {\bibfield  {journal} {\bibinfo  {journal}
  {Quantum}\ }\textbf {\bibinfo {volume} {1}},\ \bibinfo {pages} {31} (\bibinfo
  {year} {2017})},\ \Eprint {http://arxiv.org/abs/1703.07847}
  {arXiv:1703.07847} \BibitemShut {NoStop}%
\bibitem [{\citenamefont {Campbell}\ and\ \citenamefont
  {Howard}(2018)}]{Campbell2018}%
  \BibitemOpen
  \bibfield  {author} {\bibinfo {author} {\bibfnamefont {E.~T.}\ \bibnamefont
  {Campbell}}\ and\ \bibinfo {author} {\bibfnamefont {M.}~\bibnamefont
  {Howard}},\ }\href {\doibase 10.22331/q-2018-03-14-56} {\bibfield  {journal}
  {\bibinfo  {journal} {Quantum}\ }\textbf {\bibinfo {volume} {2}},\ \bibinfo
  {pages} {56} (\bibinfo {year} {2018})},\ \Eprint
  {http://arxiv.org/abs/1709.02214} {arXiv:1709.02214} \BibitemShut {NoStop}%
\bibitem [{\citenamefont {Hastings}\ and\ \citenamefont
  {Haah}(2018)}]{Hastings2018}%
  \BibitemOpen
  \bibfield  {author} {\bibinfo {author} {\bibfnamefont {M.~B.}\ \bibnamefont
  {Hastings}}\ and\ \bibinfo {author} {\bibfnamefont {J.}~\bibnamefont
  {Haah}},\ }\href {\doibase 10.1103/PhysRevLett.120.050504} {\bibfield
  {journal} {\bibinfo  {journal} {Physical Review Letters}\ }\textbf {\bibinfo
  {volume} {120}},\ \bibinfo {pages} {050504} (\bibinfo {year} {2018})},\
  \Eprint {http://arxiv.org/abs/1709.03543} {arXiv:1709.03543} \BibitemShut
  {NoStop}%
\bibitem [{\citenamefont {Krishna}\ and\ \citenamefont
  {Tillich}(2018)}]{Krishna2018a}%
  \BibitemOpen
  \bibfield  {author} {\bibinfo {author} {\bibfnamefont {A.}~\bibnamefont
  {Krishna}}\ and\ \bibinfo {author} {\bibfnamefont {J.-P.}\ \bibnamefont
  {Tillich}},\ }\href {http://arxiv.org/abs/1811.08461} {\  (\bibinfo {year}
  {2018})},\ \Eprint {http://arxiv.org/abs/1811.08461} {arXiv:1811.08461}
  \BibitemShut {NoStop}%
\bibitem [{\citenamefont {Chamberland}\ and\ \citenamefont
  {Cross}(2018)}]{Chamberland2018}%
  \BibitemOpen
  \bibfield  {author} {\bibinfo {author} {\bibfnamefont {C.}~\bibnamefont
  {Chamberland}}\ and\ \bibinfo {author} {\bibfnamefont {A.}~\bibnamefont
  {Cross}},\ }\href@noop {} {\  (\bibinfo {year} {2018})},\ \Eprint
  {http://arxiv.org/abs/1811.00566v1} {arXiv:1811.00566v1} \BibitemShut
  {NoStop}%
\bibitem [{\citenamefont {Chitambar}\ and\ \citenamefont
  {Gour}(2019)}]{Chitambar2018}%
  \BibitemOpen
  \bibfield  {author} {\bibinfo {author} {\bibfnamefont {E.}~\bibnamefont
  {Chitambar}}\ and\ \bibinfo {author} {\bibfnamefont {G.}~\bibnamefont
  {Gour}},\ }\href {\doibase 10.1103/RevModPhys.91.025001} {\bibfield
  {journal} {\bibinfo  {journal} {Reviews of Modern Physics}\ }\textbf
  {\bibinfo {volume} {91}},\ \bibinfo {pages} {025001} (\bibinfo {year}
  {2019})},\ \bibinfo {note} {arXiv:1806.06107}\BibitemShut {NoStop}%
\bibitem [{\citenamefont {Veitch}\ \emph {et~al.}(2012)\citenamefont {Veitch},
  \citenamefont {Ferrie}, \citenamefont {Gross},\ and\ \citenamefont
  {Emerson}}]{Veitch2012}%
  \BibitemOpen
  \bibfield  {author} {\bibinfo {author} {\bibfnamefont {V.}~\bibnamefont
  {Veitch}}, \bibinfo {author} {\bibfnamefont {C.}~\bibnamefont {Ferrie}},
  \bibinfo {author} {\bibfnamefont {D.}~\bibnamefont {Gross}}, \ and\ \bibinfo
  {author} {\bibfnamefont {J.}~\bibnamefont {Emerson}},\ }\href {\doibase
  10.1088/1367-2630/14/11/113011} {\bibfield  {journal} {\bibinfo  {journal}
  {New Journal of Physics}\ }\textbf {\bibinfo {volume} {14}},\ \bibinfo
  {pages} {113011} (\bibinfo {year} {2012})},\ \Eprint
  {http://arxiv.org/abs/1201.1256} {arXiv:1201.1256} \BibitemShut {NoStop}%
\bibitem [{\citenamefont {Veitch}\ \emph {et~al.}(2014)\citenamefont {Veitch},
  \citenamefont {{Hamed Mousavian}}, \citenamefont {Gottesman},\ and\
  \citenamefont {Emerson}}]{Veitch2014}%
  \BibitemOpen
  \bibfield  {author} {\bibinfo {author} {\bibfnamefont {V.}~\bibnamefont
  {Veitch}}, \bibinfo {author} {\bibfnamefont {S.~A.}\ \bibnamefont {{Hamed
  Mousavian}}}, \bibinfo {author} {\bibfnamefont {D.}~\bibnamefont
  {Gottesman}}, \ and\ \bibinfo {author} {\bibfnamefont {J.}~\bibnamefont
  {Emerson}},\ }\href {\doibase 10.1088/1367-2630/16/1/013009} {\bibfield
  {journal} {\bibinfo  {journal} {New Journal of Physics}\ }\textbf {\bibinfo
  {volume} {16}},\ \bibinfo {pages} {013009} (\bibinfo {year} {2014})},\
  \Eprint {http://arxiv.org/abs/1307.7171} {arXiv:1307.7171} \BibitemShut
  {NoStop}%
\bibitem [{\citenamefont {Howard}\ \emph {et~al.}(2014)\citenamefont {Howard},
  \citenamefont {Wallman}, \citenamefont {Veitch},\ and\ \citenamefont
  {Emerson}}]{Howard2014}%
  \BibitemOpen
  \bibfield  {author} {\bibinfo {author} {\bibfnamefont {M.}~\bibnamefont
  {Howard}}, \bibinfo {author} {\bibfnamefont {J.~J.}\ \bibnamefont {Wallman}},
  \bibinfo {author} {\bibfnamefont {V.}~\bibnamefont {Veitch}}, \ and\ \bibinfo
  {author} {\bibfnamefont {J.}~\bibnamefont {Emerson}},\ }\href {\doibase
  10.1038/nature13460} {\bibfield  {journal} {\bibinfo  {journal} {Nature}\
  }\textbf {\bibinfo {volume} {510}},\ \bibinfo {pages} {351} (\bibinfo {year}
  {2014})},\ \Eprint {http://arxiv.org/abs/1401.4174} {arXiv:1401.4174}
  \BibitemShut {NoStop}%
\bibitem [{\citenamefont {Delfosse}\ \emph {et~al.}(2015)\citenamefont
  {Delfosse}, \citenamefont {{Allard Guerin}}, \citenamefont {Bian},\ and\
  \citenamefont {Raussendorf}}]{Delfosse2015}%
  \BibitemOpen
  \bibfield  {author} {\bibinfo {author} {\bibfnamefont {N.}~\bibnamefont
  {Delfosse}}, \bibinfo {author} {\bibfnamefont {P.}~\bibnamefont {{Allard
  Guerin}}}, \bibinfo {author} {\bibfnamefont {J.}~\bibnamefont {Bian}}, \ and\
  \bibinfo {author} {\bibfnamefont {R.}~\bibnamefont {Raussendorf}},\ }\href
  {\doibase 10.1103/PhysRevX.5.021003} {\bibfield  {journal} {\bibinfo
  {journal} {Physical Review X}\ }\textbf {\bibinfo {volume} {5}},\ \bibinfo
  {pages} {021003} (\bibinfo {year} {2015})},\ \Eprint
  {http://arxiv.org/abs/1409.5170} {arXiv:1409.5170} \BibitemShut {NoStop}%
\bibitem [{\citenamefont {Mari}\ and\ \citenamefont {Eisert}(2012)}]{Mari2012}%
  \BibitemOpen
  \bibfield  {author} {\bibinfo {author} {\bibfnamefont {A.}~\bibnamefont
  {Mari}}\ and\ \bibinfo {author} {\bibfnamefont {J.}~\bibnamefont {Eisert}},\
  }\href {\doibase 10.1103/PhysRevLett.109.230503} {\bibfield  {journal}
  {\bibinfo  {journal} {Physical Review Letters}\ }\textbf {\bibinfo {volume}
  {109}},\ \bibinfo {pages} {230503} (\bibinfo {year} {2012})},\ \Eprint
  {http://arxiv.org/abs/1208.3660} {arXiv:1208.3660} \BibitemShut {NoStop}%
\bibitem [{\citenamefont {Howard}\ and\ \citenamefont
  {Campbell}(2017)}]{Howard2016}%
  \BibitemOpen
  \bibfield  {author} {\bibinfo {author} {\bibfnamefont {M.}~\bibnamefont
  {Howard}}\ and\ \bibinfo {author} {\bibfnamefont {E.}~\bibnamefont
  {Campbell}},\ }\href {\doibase 10.1103/PhysRevLett.118.090501} {\bibfield
  {journal} {\bibinfo  {journal} {Physical Review Letters}\ }\textbf {\bibinfo
  {volume} {118}},\ \bibinfo {pages} {090501} (\bibinfo {year} {2017})},\
  \Eprint {http://arxiv.org/abs/1609.07488} {arXiv:1609.07488} \BibitemShut
  {NoStop}%
\bibitem [{\citenamefont {Bravyi}\ \emph {et~al.}(2016)\citenamefont {Bravyi},
  \citenamefont {Smith},\ and\ \citenamefont {Smolin}}]{Bravyi2016}%
  \BibitemOpen
  \bibfield  {author} {\bibinfo {author} {\bibfnamefont {S.}~\bibnamefont
  {Bravyi}}, \bibinfo {author} {\bibfnamefont {G.}~\bibnamefont {Smith}}, \
  and\ \bibinfo {author} {\bibfnamefont {J.~A.}\ \bibnamefont {Smolin}},\
  }\href {\doibase 10.1103/PhysRevX.6.021043} {\bibfield  {journal} {\bibinfo
  {journal} {Physical Review X}\ }\textbf {\bibinfo {volume} {6}},\ \bibinfo
  {pages} {021043} (\bibinfo {year} {2016})}\BibitemShut {NoStop}%
\bibitem [{\citenamefont {Heinrich}\ and\ \citenamefont
  {Gross}(2018)}]{Heinrich2018}%
  \BibitemOpen
  \bibfield  {author} {\bibinfo {author} {\bibfnamefont {M.}~\bibnamefont
  {Heinrich}}\ and\ \bibinfo {author} {\bibfnamefont {D.}~\bibnamefont
  {Gross}},\ }\href {http://arxiv.org/abs/1807.10296} {\  (\bibinfo {year}
  {2018})},\ \Eprint {http://arxiv.org/abs/1807.10296} {arXiv:1807.10296}
  \BibitemShut {NoStop}%
\bibitem [{\citenamefont {Wootters}(1987)}]{Wootters1987}%
  \BibitemOpen
  \bibfield  {author} {\bibinfo {author} {\bibfnamefont {W.~K.}\ \bibnamefont
  {Wootters}},\ }\href {\doibase 10.1016/0003-4916(87)90176-X} {\bibfield
  {journal} {\bibinfo  {journal} {Annals of Physics}\ }\textbf {\bibinfo
  {volume} {176}},\ \bibinfo {pages} {1} (\bibinfo {year} {1987})}\BibitemShut
  {NoStop}%
\bibitem [{\citenamefont {Gross}(2006)}]{Gross2006}%
  \BibitemOpen
  \bibfield  {author} {\bibinfo {author} {\bibfnamefont {D.}~\bibnamefont
  {Gross}},\ }\href {\doibase 10.1063/1.2393152} {\bibfield  {journal}
  {\bibinfo  {journal} {Journal of Mathematical Physics}\ }\textbf {\bibinfo
  {volume} {47}},\ \bibinfo {pages} {122107} (\bibinfo {year} {2006})},\
  \Eprint {http://arxiv.org/abs/quant-ph/0602001} {arXiv:quant-ph/0602001
  [quant-ph]} \BibitemShut {NoStop}%
\bibitem [{\citenamefont {Gross}(2007)}]{Gross2007}%
  \BibitemOpen
  \bibfield  {author} {\bibinfo {author} {\bibfnamefont {D.}~\bibnamefont
  {Gross}},\ }\href {\doibase 10.1007/s00340-006-2510-9} {\bibfield  {journal}
  {\bibinfo  {journal} {Applied Physics B}\ }\textbf {\bibinfo {volume} {86}},\
  \bibinfo {pages} {367} (\bibinfo {year} {2007})},\ \Eprint
  {http://arxiv.org/abs/quant-ph/0702004} {arXiv:quant-ph/0702004 [quant-ph]}
  \BibitemShut {NoStop}%
\bibitem [{Note1()}]{Note1}%
  \BibitemOpen
  \bibinfo {note} {See Supplemental Material [url] for detailed mathematical
  proofs and developments for the assertions in the main text. The Supplemental
  Material includes Refs.~\cite
  {PhysRevA.71.042315,Ferrie2011,Fawzi2017a,Tomamichel2009,Mosonyi2011,Uhlmann1976,FL13,Beigi2013,Wilde2018a}}\BibitemShut
  {NoStop}%
\bibitem [{\citenamefont {Peres}(1996)}]{Peres1996}%
  \BibitemOpen
  \bibfield  {author} {\bibinfo {author} {\bibfnamefont {A.}~\bibnamefont
  {Peres}},\ }\href {\doibase 10.1103/PhysRevLett.77.1413} {\bibfield
  {journal} {\bibinfo  {journal} {Physical Review Letters}\ }\textbf {\bibinfo
  {volume} {77}},\ \bibinfo {pages} {1413} (\bibinfo {year}
  {1996})}\BibitemShut {NoStop}%
\bibitem [{\citenamefont {Horodecki}\ \emph {et~al.}(1998)\citenamefont
  {Horodecki}, \citenamefont {Horodecki},\ and\ \citenamefont
  {Horodecki}}]{Horodecki1998}%
  \BibitemOpen
  \bibfield  {author} {\bibinfo {author} {\bibfnamefont {M.}~\bibnamefont
  {Horodecki}}, \bibinfo {author} {\bibfnamefont {P.}~\bibnamefont
  {Horodecki}}, \ and\ \bibinfo {author} {\bibfnamefont {R.}~\bibnamefont
  {Horodecki}},\ }\href {\doibase 10.1103/PhysRevLett.80.5239} {\bibfield
  {journal} {\bibinfo  {journal} {Physical Review Letters}\ }\textbf {\bibinfo
  {volume} {80}},\ \bibinfo {pages} {5239} (\bibinfo {year}
  {1998})}\BibitemShut {NoStop}%
\bibitem [{\citenamefont {Rains}(2001)}]{Rains2001}%
  \BibitemOpen
  \bibfield  {author} {\bibinfo {author} {\bibfnamefont {E.~M.}\ \bibnamefont
  {Rains}},\ }\href {\doibase 10.1109/18.959270} {\bibfield  {journal}
  {\bibinfo  {journal} {IEEE Transactions on Information Theory}\ }\textbf
  {\bibinfo {volume} {47}},\ \bibinfo {pages} {2921} (\bibinfo {year}
  {2001})}\BibitemShut {NoStop}%
\bibitem [{\citenamefont {Wang}\ and\ \citenamefont
  {Duan}(2016{\natexlab{a}})}]{Wang2016}%
  \BibitemOpen
  \bibfield  {author} {\bibinfo {author} {\bibfnamefont {X.}~\bibnamefont
  {Wang}}\ and\ \bibinfo {author} {\bibfnamefont {R.}~\bibnamefont {Duan}},\
  }\href {\doibase 10.1103/PhysRevA.94.050301} {\bibfield  {journal} {\bibinfo
  {journal} {Physical Review A}\ }\textbf {\bibinfo {volume} {94}},\ \bibinfo
  {pages} {050301} (\bibinfo {year} {2016}{\natexlab{a}})}\BibitemShut
  {NoStop}%
\bibitem [{\citenamefont {Wang}\ and\ \citenamefont {Duan}(2017)}]{Wang2017e}%
  \BibitemOpen
  \bibfield  {author} {\bibinfo {author} {\bibfnamefont {X.}~\bibnamefont
  {Wang}}\ and\ \bibinfo {author} {\bibfnamefont {R.}~\bibnamefont {Duan}},\
  }\href {\doibase 10.1103/PhysRevA.95.062322} {\bibfield  {journal} {\bibinfo
  {journal} {Physical Review A}\ }\textbf {\bibinfo {volume} {95}},\ \bibinfo
  {pages} {062322} (\bibinfo {year} {2017})}\BibitemShut {NoStop}%
\bibitem [{\citenamefont {Fang}\ \emph {et~al.}(2017)\citenamefont {Fang},
  \citenamefont {Wang}, \citenamefont {Tomamichel},\ and\ \citenamefont
  {Duan}}]{Fang2017}%
  \BibitemOpen
  \bibfield  {author} {\bibinfo {author} {\bibfnamefont {K.}~\bibnamefont
  {Fang}}, \bibinfo {author} {\bibfnamefont {X.}~\bibnamefont {Wang}}, \bibinfo
  {author} {\bibfnamefont {M.}~\bibnamefont {Tomamichel}}, \ and\ \bibinfo
  {author} {\bibfnamefont {R.}~\bibnamefont {Duan}},\ }\href@noop {} {\bibfield
   {journal} {\bibinfo  {journal} {arXiv:1706.06221}\ } (\bibinfo {year}
  {2017})},\ \Eprint {http://arxiv.org/abs/1706.06221} {arXiv:1706.06221}
  \BibitemShut {NoStop}%
\bibitem [{\citenamefont {Tomamichel}\ \emph {et~al.}(2017)\citenamefont
  {Tomamichel}, \citenamefont {Wilde},\ and\ \citenamefont
  {Winter}}]{Tomamichel2015a}%
  \BibitemOpen
  \bibfield  {author} {\bibinfo {author} {\bibfnamefont {M.}~\bibnamefont
  {Tomamichel}}, \bibinfo {author} {\bibfnamefont {M.~M.}\ \bibnamefont
  {Wilde}}, \ and\ \bibinfo {author} {\bibfnamefont {A.}~\bibnamefont
  {Winter}},\ }\href {\doibase 10.1109/TIT.2016.2615847} {\bibfield  {journal}
  {\bibinfo  {journal} {IEEE Transactions on Information Theory}\ }\textbf
  {\bibinfo {volume} {63}},\ \bibinfo {pages} {715} (\bibinfo {year}
  {2017})}\BibitemShut {NoStop}%
\bibitem [{\citenamefont {Tomamichel}\ \emph {et~al.}(2016)\citenamefont
  {Tomamichel}, \citenamefont {Berta},\ and\ \citenamefont
  {Renes}}]{Tomamichel2016}%
  \BibitemOpen
  \bibfield  {author} {\bibinfo {author} {\bibfnamefont {M.}~\bibnamefont
  {Tomamichel}}, \bibinfo {author} {\bibfnamefont {M.}~\bibnamefont {Berta}}, \
  and\ \bibinfo {author} {\bibfnamefont {J.~M.}\ \bibnamefont {Renes}},\ }\href
  {\doibase 10.1038/ncomms11419} {\bibfield  {journal} {\bibinfo  {journal}
  {Nature Communications}\ }\textbf {\bibinfo {volume} {7}},\ \bibinfo {pages}
  {11419} (\bibinfo {year} {2016})}\BibitemShut {NoStop}%
\bibitem [{\citenamefont {Wang}\ and\ \citenamefont
  {Duan}(2016{\natexlab{b}})}]{Wang2016a}%
  \BibitemOpen
  \bibfield  {author} {\bibinfo {author} {\bibfnamefont {X.}~\bibnamefont
  {Wang}}\ and\ \bibinfo {author} {\bibfnamefont {R.}~\bibnamefont {Duan}},\
  }in\ \href {\doibase 10.1109/ISIT.2016.7541587} {\emph {\bibinfo {booktitle}
  {2016 IEEE International Symposium on Information Theory (ISIT)}}},\ Vol.\
  \bibinfo {volume} {2016-Augus}\ (\bibinfo  {publisher} {IEEE},\ \bibinfo
  {year} {2016})\ pp.\ \bibinfo {pages} {1690--1694}\BibitemShut {NoStop}%
\bibitem [{\citenamefont {Wang}\ \emph {et~al.}(2018)\citenamefont {Wang},
  \citenamefont {Fang},\ and\ \citenamefont {Duan}}]{Wang2017d}%
  \BibitemOpen
  \bibfield  {author} {\bibinfo {author} {\bibfnamefont {X.}~\bibnamefont
  {Wang}}, \bibinfo {author} {\bibfnamefont {K.}~\bibnamefont {Fang}}, \ and\
  \bibinfo {author} {\bibfnamefont {R.}~\bibnamefont {Duan}},\ }\href {\doibase
  10.1109/TIT.2018.2874031} {\bibfield  {journal} {\bibinfo  {journal} {IEEE
  Transactions on Information Theory}\ } (\bibinfo {year} {2018}),\
  10.1109/TIT.2018.2874031},\ \Eprint {http://arxiv.org/abs/1709.00200}
  {arXiv:1709.00200} \BibitemShut {NoStop}%
\bibitem [{\citenamefont {Audenaert}\ \emph {et~al.}(2002)\citenamefont
  {Audenaert}, \citenamefont {{De Moor}}, \citenamefont {Vollbrecht},\ and\
  \citenamefont {Werner}}]{Audenaert2002}%
  \BibitemOpen
  \bibfield  {author} {\bibinfo {author} {\bibfnamefont {K.}~\bibnamefont
  {Audenaert}}, \bibinfo {author} {\bibfnamefont {B.}~\bibnamefont {{De
  Moor}}}, \bibinfo {author} {\bibfnamefont {K.~G.~H.}\ \bibnamefont
  {Vollbrecht}}, \ and\ \bibinfo {author} {\bibfnamefont {R.~F.}\ \bibnamefont
  {Werner}},\ }\href {\doibase 10.1103/PhysRevA.66.032310} {\bibfield
  {journal} {\bibinfo  {journal} {Physical Review A}\ }\textbf {\bibinfo
  {volume} {66}},\ \bibinfo {pages} {032310} (\bibinfo {year}
  {2002})}\BibitemShut {NoStop}%
\bibitem [{\citenamefont {Vidal}\ and\ \citenamefont
  {Werner}(2002)}]{Vidal2002}%
  \BibitemOpen
  \bibfield  {author} {\bibinfo {author} {\bibfnamefont {G.}~\bibnamefont
  {Vidal}}\ and\ \bibinfo {author} {\bibfnamefont {R.~F.}\ \bibnamefont
  {Werner}},\ }\href {\doibase 10.1103/PhysRevA.65.032314} {\bibfield
  {journal} {\bibinfo  {journal} {Physical Review A}\ }\textbf {\bibinfo
  {volume} {65}},\ \bibinfo {pages} {032314} (\bibinfo {year}
  {2002})}\BibitemShut {NoStop}%
\bibitem [{\citenamefont {Plenio}(2005)}]{Plenio2005b}%
  \BibitemOpen
  \bibfield  {author} {\bibinfo {author} {\bibfnamefont {M.~B.}\ \bibnamefont
  {Plenio}},\ }\href {\doibase 10.1103/PhysRevLett.95.090503} {\bibfield
  {journal} {\bibinfo  {journal} {Physical Review Letters}\ }\textbf {\bibinfo
  {volume} {95}},\ \bibinfo {pages} {090503} (\bibinfo {year}
  {2005})}\BibitemShut {NoStop}%
\bibitem [{Note2()}]{Note2}%
  \BibitemOpen
  \bibinfo {note} {Greek for ``wonder'' or ``marvel''}\BibitemShut {NoStop}%
\bibitem [{\citenamefont {Umegaki}(1962)}]{Umegaki1962}%
  \BibitemOpen
  \bibfield  {author} {\bibinfo {author} {\bibfnamefont {H.}~\bibnamefont
  {Umegaki}},\ }\href {\doibase 10.2996/kmj/1138844604} {\bibfield  {journal}
  {\bibinfo  {journal} {Kodai Mathematical Seminar Reports}\ }\textbf {\bibinfo
  {volume} {14}},\ \bibinfo {pages} {59} (\bibinfo {year} {1962})}\BibitemShut
  {NoStop}%
\bibitem [{\citenamefont {Polyanskiy}\ and\ \citenamefont
  {Verdu}(2010)}]{Polyanskiy2010b}%
  \BibitemOpen
  \bibfield  {author} {\bibinfo {author} {\bibfnamefont {Y.}~\bibnamefont
  {Polyanskiy}}\ and\ \bibinfo {author} {\bibfnamefont {S.}~\bibnamefont
  {Verdu}},\ }in\ \href {\doibase 10.1109/ALLERTON.2010.5707067} {\emph
  {\bibinfo {booktitle} {2010 48th Annual Allerton Conference on Communication,
  Control, and Computing (Allerton)}}}\ (\bibinfo  {publisher} {IEEE},\
  \bibinfo {year} {2010})\ pp.\ \bibinfo {pages} {1327--1333}\BibitemShut
  {NoStop}%
\bibitem [{\citenamefont {Sharma}\ and\ \citenamefont {Warsi}(2012)}]{SW12}%
  \BibitemOpen
  \bibfield  {author} {\bibinfo {author} {\bibfnamefont {N.}~\bibnamefont
  {Sharma}}\ and\ \bibinfo {author} {\bibfnamefont {N.~A.}\ \bibnamefont
  {Warsi}},\ }\href@noop {} {\  (\bibinfo {year} {2012})},\ \Eprint
  {http://arxiv.org/abs/arXiv:1205.1712} {arXiv:1205.1712} \BibitemShut
  {NoStop}%
\bibitem [{\citenamefont {Petz}(1986)}]{P86}%
  \BibitemOpen
  \bibfield  {author} {\bibinfo {author} {\bibfnamefont {D.}~\bibnamefont
  {Petz}},\ }\href {\doibase 10.1016/0034-4877(86)90067-4} {\bibfield
  {journal} {\bibinfo  {journal} {Reports in Mathematical Physics}\ }\textbf
  {\bibinfo {volume} {23}},\ \bibinfo {pages} {57} (\bibinfo {year}
  {1986})}\BibitemShut {NoStop}%
\bibitem [{\citenamefont {M{\"{u}}ller-Lennert}\ \emph
  {et~al.}(2013)\citenamefont {M{\"{u}}ller-Lennert}, \citenamefont {Dupuis},
  \citenamefont {Szehr}, \citenamefont {Fehr},\ and\ \citenamefont
  {Tomamichel}}]{Muller-Lennert2013}%
  \BibitemOpen
  \bibfield  {author} {\bibinfo {author} {\bibfnamefont {M.}~\bibnamefont
  {M{\"{u}}ller-Lennert}}, \bibinfo {author} {\bibfnamefont {F.}~\bibnamefont
  {Dupuis}}, \bibinfo {author} {\bibfnamefont {O.}~\bibnamefont {Szehr}},
  \bibinfo {author} {\bibfnamefont {S.}~\bibnamefont {Fehr}}, \ and\ \bibinfo
  {author} {\bibfnamefont {M.}~\bibnamefont {Tomamichel}},\ }\href {\doibase
  10.1063/1.4838856} {\bibfield  {journal} {\bibinfo  {journal} {Journal of
  Mathematical Physics}\ }\textbf {\bibinfo {volume} {54}},\ \bibinfo {pages}
  {122203} (\bibinfo {year} {2013})},\ \Eprint {http://arxiv.org/abs/1306.3142}
  {arXiv:1306.3142} \BibitemShut {NoStop}%
\bibitem [{\citenamefont {Wilde}\ \emph {et~al.}(2014)\citenamefont {Wilde},
  \citenamefont {Winter},\ and\ \citenamefont {Yang}}]{Wilde2014a}%
  \BibitemOpen
  \bibfield  {author} {\bibinfo {author} {\bibfnamefont {M.~M.}\ \bibnamefont
  {Wilde}}, \bibinfo {author} {\bibfnamefont {A.}~\bibnamefont {Winter}}, \
  and\ \bibinfo {author} {\bibfnamefont {D.}~\bibnamefont {Yang}},\ }\href
  {\doibase 10.1007/s00220-014-2122-x} {\bibfield  {journal} {\bibinfo
  {journal} {Communications in Mathematical Physics}\ }\textbf {\bibinfo
  {volume} {331}},\ \bibinfo {pages} {593} (\bibinfo {year}
  {2014})}\BibitemShut {NoStop}%
\bibitem [{\citenamefont {Datta}(2009)}]{Datta2009}%
  \BibitemOpen
  \bibfield  {author} {\bibinfo {author} {\bibfnamefont {N.}~\bibnamefont
  {Datta}},\ }\href {\doibase 10.1109/TIT.2009.2018325} {\bibfield  {journal}
  {\bibinfo  {journal} {IEEE Transactions on Information Theory}\ }\textbf
  {\bibinfo {volume} {55}},\ \bibinfo {pages} {2816} (\bibinfo {year}
  {2009})}\BibitemShut {NoStop}%
\bibitem [{\citenamefont {Boyd}\ and\ \citenamefont
  {Vandenberghe}(2004)}]{Boyd2004}%
  \BibitemOpen
  \bibfield  {author} {\bibinfo {author} {\bibfnamefont {S.}~\bibnamefont
  {Boyd}}\ and\ \bibinfo {author} {\bibfnamefont {L.}~\bibnamefont
  {Vandenberghe}},\ }\href@noop {} {\emph {\bibinfo {title} {{Convex
  optimization}}}}\ (\bibinfo  {publisher} {Cambridge University Press},\
  \bibinfo {year} {2004})\BibitemShut {NoStop}%
\bibitem [{\citenamefont {Bravyi}\ \emph {et~al.}(2018)\citenamefont {Bravyi},
  \citenamefont {Browne}, \citenamefont {Calpin}, \citenamefont {Campbell},
  \citenamefont {Gosset},\ and\ \citenamefont {Howard}}]{Bravyi2018}%
  \BibitemOpen
  \bibfield  {author} {\bibinfo {author} {\bibfnamefont {S.}~\bibnamefont
  {Bravyi}}, \bibinfo {author} {\bibfnamefont {D.}~\bibnamefont {Browne}},
  \bibinfo {author} {\bibfnamefont {P.}~\bibnamefont {Calpin}}, \bibinfo
  {author} {\bibfnamefont {E.}~\bibnamefont {Campbell}}, \bibinfo {author}
  {\bibfnamefont {D.}~\bibnamefont {Gosset}}, \ and\ \bibinfo {author}
  {\bibfnamefont {M.}~\bibnamefont {Howard}},\ }\href
  {http://arxiv.org/abs/1808.00128} {\  (\bibinfo {year} {2018})},\ \Eprint
  {http://arxiv.org/abs/1808.00128} {arXiv:1808.00128} \BibitemShut {NoStop}%
\bibitem [{\citenamefont {Howard}\ and\ \citenamefont
  {Vala}(2012)}]{Howard2012}%
  \BibitemOpen
  \bibfield  {author} {\bibinfo {author} {\bibfnamefont {M.}~\bibnamefont
  {Howard}}\ and\ \bibinfo {author} {\bibfnamefont {J.}~\bibnamefont {Vala}},\
  }\href {\doibase 10.1103/PhysRevA.86.022316} {\bibfield  {journal} {\bibinfo
  {journal} {Physical Review A}\ }\textbf {\bibinfo {volume} {86}},\ \bibinfo
  {pages} {022316} (\bibinfo {year} {2012})},\ \Eprint
  {http://arxiv.org/abs/1206.1598} {arXiv:1206.1598} \BibitemShut {NoStop}%
\bibitem [{\citenamefont {Anwar}\ \emph {et~al.}(2012)\citenamefont {Anwar},
  \citenamefont {Campbell},\ and\ \citenamefont {Browne}}]{Anwar2012}%
  \BibitemOpen
  \bibfield  {author} {\bibinfo {author} {\bibfnamefont {H.}~\bibnamefont
  {Anwar}}, \bibinfo {author} {\bibfnamefont {E.~T.}\ \bibnamefont {Campbell}},
  \ and\ \bibinfo {author} {\bibfnamefont {D.~E.}\ \bibnamefont {Browne}},\
  }\href {\doibase 10.1088/1367-2630/14/6/063006} {\bibfield  {journal}
  {\bibinfo  {journal} {New Journal of Physics}\ }\textbf {\bibinfo {volume}
  {14}},\ \bibinfo {pages} {063006} (\bibinfo {year} {2012})},\ \Eprint
  {http://arxiv.org/abs/1202.2326} {arXiv:1202.2326} \BibitemShut {NoStop}%
\bibitem [{\citenamefont {Campbell}\ \emph {et~al.}(2012)\citenamefont
  {Campbell}, \citenamefont {Anwar},\ and\ \citenamefont
  {Browne}}]{Campbell2012b}%
  \BibitemOpen
  \bibfield  {author} {\bibinfo {author} {\bibfnamefont {E.~T.}\ \bibnamefont
  {Campbell}}, \bibinfo {author} {\bibfnamefont {H.}~\bibnamefont {Anwar}}, \
  and\ \bibinfo {author} {\bibfnamefont {D.~E.}\ \bibnamefont {Browne}},\
  }\href {\doibase 10.1103/PhysRevX.2.041021} {\bibfield  {journal} {\bibinfo
  {journal} {Physical Review X}\ }\textbf {\bibinfo {volume} {2}},\ \bibinfo
  {pages} {041021} (\bibinfo {year} {2012})},\ \Eprint
  {http://arxiv.org/abs/1205.3104} {arXiv:1205.3104} \BibitemShut {NoStop}%
\bibitem [{\citenamefont {Campbell}(2014)}]{Campbell2014}%
  \BibitemOpen
  \bibfield  {author} {\bibinfo {author} {\bibfnamefont {E.~T.}\ \bibnamefont
  {Campbell}},\ }\href {\doibase 10.1103/PhysRevLett.113.230501} {\bibfield
  {journal} {\bibinfo  {journal} {Physical Review Letters}\ }\textbf {\bibinfo
  {volume} {113}},\ \bibinfo {pages} {230501} (\bibinfo {year} {2014})},\
  \Eprint {http://arxiv.org/abs/1406.3055} {arXiv:1406.3055} \BibitemShut
  {NoStop}%
\bibitem [{\citenamefont {Dawkins}\ and\ \citenamefont
  {Howard}(2015)}]{Dawkins2015}%
  \BibitemOpen
  \bibfield  {author} {\bibinfo {author} {\bibfnamefont {H.}~\bibnamefont
  {Dawkins}}\ and\ \bibinfo {author} {\bibfnamefont {M.}~\bibnamefont
  {Howard}},\ }\href {\doibase 10.1103/PhysRevLett.115.030501} {\bibfield
  {journal} {\bibinfo  {journal} {Physical Review Letters}\ }\textbf {\bibinfo
  {volume} {115}},\ \bibinfo {pages} {030501} (\bibinfo {year} {2015})},\
  \Eprint {http://arxiv.org/abs/1504.05965} {arXiv:1504.05965} \BibitemShut
  {NoStop}%
\bibitem [{\citenamefont {Rains}(1999)}]{Rains1999}%
  \BibitemOpen
  \bibfield  {author} {\bibinfo {author} {\bibfnamefont {E.~M.}\ \bibnamefont
  {Rains}},\ }\href {\doibase 10.1103/PhysRevA.60.179} {\bibfield  {journal}
  {\bibinfo  {journal} {Physical Review A}\ }\textbf {\bibinfo {volume} {60}},\
  \bibinfo {pages} {179} (\bibinfo {year} {1999})}\BibitemShut {NoStop}%
\bibitem [{\citenamefont {Buscemi}\ and\ \citenamefont
  {Datta}(2010)}]{Buscemi2010}%
  \BibitemOpen
  \bibfield  {author} {\bibinfo {author} {\bibfnamefont {F.}~\bibnamefont
  {Buscemi}}\ and\ \bibinfo {author} {\bibfnamefont {N.}~\bibnamefont
  {Datta}},\ }\href {\doibase 10.1109/TIT.2009.2039166} {\bibfield  {journal}
  {\bibinfo  {journal} {IEEE Transactions on Information Theory}\ }\textbf
  {\bibinfo {volume} {56}},\ \bibinfo {pages} {1447} (\bibinfo {year}
  {2010})}\BibitemShut {NoStop}%
\bibitem [{\citenamefont {Wang}\ and\ \citenamefont {Renner}(2012)}]{Wang2012}%
  \BibitemOpen
  \bibfield  {author} {\bibinfo {author} {\bibfnamefont {L.}~\bibnamefont
  {Wang}}\ and\ \bibinfo {author} {\bibfnamefont {R.}~\bibnamefont {Renner}},\
  }\href {\doibase 10.1103/PhysRevLett.108.200501} {\bibfield  {journal}
  {\bibinfo  {journal} {Physical Review Letters}\ }\textbf {\bibinfo {volume}
  {108}},\ \bibinfo {pages} {200501} (\bibinfo {year} {2012})}\BibitemShut
  {NoStop}%
\bibitem [{\citenamefont {Hiai}\ and\ \citenamefont {Petz}(1991)}]{Hiai1991}%
  \BibitemOpen
  \bibfield  {author} {\bibinfo {author} {\bibfnamefont {F.}~\bibnamefont
  {Hiai}}\ and\ \bibinfo {author} {\bibfnamefont {D.}~\bibnamefont {Petz}},\
  }\href {\doibase 10.1007/BF02100287} {\bibfield  {journal} {\bibinfo
  {journal} {Communications in Mathematical Physics}\ }\textbf {\bibinfo
  {volume} {143}},\ \bibinfo {pages} {99} (\bibinfo {year} {1991})}\BibitemShut
  {NoStop}%
\bibitem [{\citenamefont {Ogawa}\ and\ \citenamefont
  {Nagaoka}(2000)}]{ogawa2000strong}%
  \BibitemOpen
  \bibfield  {author} {\bibinfo {author} {\bibfnamefont {T.}~\bibnamefont
  {Ogawa}}\ and\ \bibinfo {author} {\bibfnamefont {H.}~\bibnamefont
  {Nagaoka}},\ }\href {\doibase 10.1142/9789812563071_0003} {\bibfield
  {journal} {\bibinfo  {journal} {IEEE Transactions on Information Theory}\
  }\textbf {\bibinfo {volume} {46}},\ \bibinfo {pages} {2428} (\bibinfo {year}
  {2000})},\ \Eprint {http://arxiv.org/abs/quant-ph/9906090} {quant-ph/9906090}
  \BibitemShut {NoStop}%
\bibitem [{\citenamefont {Tomamichel}\ and\ \citenamefont
  {Hayashi}(2013)}]{Tomamichel2013a}%
  \BibitemOpen
  \bibfield  {author} {\bibinfo {author} {\bibfnamefont {M.}~\bibnamefont
  {Tomamichel}}\ and\ \bibinfo {author} {\bibfnamefont {M.}~\bibnamefont
  {Hayashi}},\ }\href {\doibase 10.1109/TIT.2013.2276628} {\bibfield  {journal}
  {\bibinfo  {journal} {IEEE Transactions on Information Theory}\ }\textbf
  {\bibinfo {volume} {59}},\ \bibinfo {pages} {7693} (\bibinfo {year}
  {2013})}\BibitemShut {NoStop}%
\bibitem [{\citenamefont {Li}(2014)}]{Li2014a}%
  \BibitemOpen
  \bibfield  {author} {\bibinfo {author} {\bibfnamefont {K.}~\bibnamefont
  {Li}},\ }\href {\doibase 10.1214/13-AOS1185} {\bibfield  {journal} {\bibinfo
  {journal} {The Annals of Statistics}\ }\textbf {\bibinfo {volume} {42}},\
  \bibinfo {pages} {171} (\bibinfo {year} {2014})}\BibitemShut {NoStop}%
\bibitem [{\citenamefont {Horodecki}\ \emph {et~al.}(2009)\citenamefont
  {Horodecki}, \citenamefont {Horodecki}, \citenamefont {Horodecki},\ and\
  \citenamefont {Horodecki}}]{Horodecki2009a}%
  \BibitemOpen
  \bibfield  {author} {\bibinfo {author} {\bibfnamefont {R.}~\bibnamefont
  {Horodecki}}, \bibinfo {author} {\bibfnamefont {P.}~\bibnamefont
  {Horodecki}}, \bibinfo {author} {\bibfnamefont {M.}~\bibnamefont
  {Horodecki}}, \ and\ \bibinfo {author} {\bibfnamefont {K.}~\bibnamefont
  {Horodecki}},\ }\href {\doibase 10.1103/RevModPhys.81.865} {\bibfield
  {journal} {\bibinfo  {journal} {Reviews of Modern Physics}\ }\textbf
  {\bibinfo {volume} {81}},\ \bibinfo {pages} {865} (\bibinfo {year}
  {2009})}\BibitemShut {NoStop}%
\bibitem [{\citenamefont {Streltsov}\ \emph {et~al.}(2017)\citenamefont
  {Streltsov}, \citenamefont {Adesso},\ and\ \citenamefont
  {Plenio}}]{Streltsov2016}%
  \BibitemOpen
  \bibfield  {author} {\bibinfo {author} {\bibfnamefont {A.}~\bibnamefont
  {Streltsov}}, \bibinfo {author} {\bibfnamefont {G.}~\bibnamefont {Adesso}}, \
  and\ \bibinfo {author} {\bibfnamefont {M.~B.}\ \bibnamefont {Plenio}},\
  }\href {\doibase 10.1103/RevModPhys.89.041003} {\bibfield  {journal}
  {\bibinfo  {journal} {Reviews of Modern Physics}\ }\textbf {\bibinfo {volume}
  {89}},\ \bibinfo {pages} {041003} (\bibinfo {year} {2017})},\ \Eprint
  {http://arxiv.org/abs/1609.02439} {arXiv:1609.02439} \BibitemShut {NoStop}%
\bibitem [{\citenamefont {Pashayan}\ \emph {et~al.}(2015)\citenamefont
  {Pashayan}, \citenamefont {Wallman},\ and\ \citenamefont {Bartlett}}]{PWB15}%
  \BibitemOpen
  \bibfield  {author} {\bibinfo {author} {\bibfnamefont {H.}~\bibnamefont
  {Pashayan}}, \bibinfo {author} {\bibfnamefont {J.~J.}\ \bibnamefont
  {Wallman}}, \ and\ \bibinfo {author} {\bibfnamefont {S.~D.}\ \bibnamefont
  {Bartlett}},\ }\href {\doibase 10.1103/PhysRevLett.115.070501} {\bibfield
  {journal} {\bibinfo  {journal} {Physical Review Letters}\ }\textbf {\bibinfo
  {volume} {115}},\ \bibinfo {pages} {070501} (\bibinfo {year} {2015})},\
  \bibinfo {note} {arXiv:1503.07525}\BibitemShut {NoStop}%
\bibitem [{\citenamefont {Brand{\~{a}}o}\ and\ \citenamefont
  {Gour}(2015)}]{Brandao2015a}%
  \BibitemOpen
  \bibfield  {author} {\bibinfo {author} {\bibfnamefont {F.~G. S.~L.}\
  \bibnamefont {Brand{\~{a}}o}}\ and\ \bibinfo {author} {\bibfnamefont
  {G.}~\bibnamefont {Gour}},\ }\href {\doibase 10.1103/PhysRevLett.115.070503}
  {\bibfield  {journal} {\bibinfo  {journal} {Physical Review Letters}\
  }\textbf {\bibinfo {volume} {115}},\ \bibinfo {pages} {070503} (\bibinfo
  {year} {2015})}\BibitemShut {NoStop}%
\bibitem [{\citenamefont {Zhuang}\ \emph {et~al.}(2018)\citenamefont {Zhuang},
  \citenamefont {Shor},\ and\ \citenamefont {Shapiro}}]{ZSS18}%
  \BibitemOpen
  \bibfield  {author} {\bibinfo {author} {\bibfnamefont {Q.}~\bibnamefont
  {Zhuang}}, \bibinfo {author} {\bibfnamefont {P.~W.}\ \bibnamefont {Shor}}, \
  and\ \bibinfo {author} {\bibfnamefont {J.~H.}\ \bibnamefont {Shapiro}},\
  }\href {\doibase 10.1103/PhysRevA.97.052317} {\bibfield  {journal} {\bibinfo
  {journal} {Physical Review A}\ }\textbf {\bibinfo {volume} {97}},\ \bibinfo
  {pages} {052317} (\bibinfo {year} {2018})},\ \bibinfo {note}
  {arXiv:1803.07580}\BibitemShut {NoStop}%
\bibitem [{\citenamefont {Takagi}\ and\ \citenamefont {Zhuang}(2018)}]{TZ18}%
  \BibitemOpen
  \bibfield  {author} {\bibinfo {author} {\bibfnamefont {R.}~\bibnamefont
  {Takagi}}\ and\ \bibinfo {author} {\bibfnamefont {Q.}~\bibnamefont
  {Zhuang}},\ }\href {\doibase 10.1103/PhysRevA.97.062337} {\bibfield
  {journal} {\bibinfo  {journal} {Physical Review A}\ }\textbf {\bibinfo
  {volume} {97}},\ \bibinfo {pages} {062337} (\bibinfo {year} {2018})},\
  \bibinfo {note} {arXiv:1804.04669}\BibitemShut {NoStop}%
\bibitem [{\citenamefont {Albarelli}\ \emph {et~al.}(2018)\citenamefont
  {Albarelli}, \citenamefont {Genoni}, \citenamefont {Paris},\ and\
  \citenamefont {Ferraro}}]{AGPF18}%
  \BibitemOpen
  \bibfield  {author} {\bibinfo {author} {\bibfnamefont {F.}~\bibnamefont
  {Albarelli}}, \bibinfo {author} {\bibfnamefont {M.~G.}\ \bibnamefont
  {Genoni}}, \bibinfo {author} {\bibfnamefont {M.~G.~A.}\ \bibnamefont
  {Paris}}, \ and\ \bibinfo {author} {\bibfnamefont {A.}~\bibnamefont
  {Ferraro}},\ }\href {\doibase 10.1103/PhysRevA.98.052350} {\bibfield
  {journal} {\bibinfo  {journal} {Physical Review A}\ }\textbf {\bibinfo
  {volume} {98}},\ \bibinfo {pages} {052350} (\bibinfo {year} {2018})},\
  \bibinfo {note} {arXiv:1804.05763}\BibitemShut {NoStop}%
\bibitem [{\citenamefont {Lloyd}\ and\ \citenamefont
  {Braunstein}(1999)}]{LB99}%
  \BibitemOpen
  \bibfield  {author} {\bibinfo {author} {\bibfnamefont {S.}~\bibnamefont
  {Lloyd}}\ and\ \bibinfo {author} {\bibfnamefont {S.~L.}\ \bibnamefont
  {Braunstein}},\ }\href {\doibase 10.1103/PhysRevLett.82.1784} {\bibfield
  {journal} {\bibinfo  {journal} {Physical Review Letters}\ }\textbf {\bibinfo
  {volume} {82}},\ \bibinfo {pages} {1784} (\bibinfo {year} {1999})},\ \bibinfo
  {note} {arXiv:quant-ph/9810082}\BibitemShut {NoStop}%
\bibitem [{\citenamefont {Gottesman}\ \emph {et~al.}(2001)\citenamefont
  {Gottesman}, \citenamefont {Kitaev},\ and\ \citenamefont {Preskill}}]{GKP01}%
  \BibitemOpen
  \bibfield  {author} {\bibinfo {author} {\bibfnamefont {D.}~\bibnamefont
  {Gottesman}}, \bibinfo {author} {\bibfnamefont {A.}~\bibnamefont {Kitaev}}, \
  and\ \bibinfo {author} {\bibfnamefont {J.}~\bibnamefont {Preskill}},\ }\href
  {\doibase 10.1103/PhysRevA.64.012310} {\bibfield  {journal} {\bibinfo
  {journal} {Physical Review A}\ }\textbf {\bibinfo {volume} {64}},\ \bibinfo
  {pages} {012310} (\bibinfo {year} {2001})},\ \bibinfo {note}
  {arXiv:quant-ph/0008040}\BibitemShut {NoStop}%
\bibitem [{\citenamefont {Hostens}\ \emph {et~al.}(2005)\citenamefont
  {Hostens}, \citenamefont {Dehaene},\ and\ \citenamefont
  {De~Moor}}]{PhysRevA.71.042315}%
  \BibitemOpen
  \bibfield  {author} {\bibinfo {author} {\bibfnamefont {E.}~\bibnamefont
  {Hostens}}, \bibinfo {author} {\bibfnamefont {J.}~\bibnamefont {Dehaene}}, \
  and\ \bibinfo {author} {\bibfnamefont {B.}~\bibnamefont {De~Moor}},\ }\href
  {\doibase 10.1103/PhysRevA.71.042315} {\bibfield  {journal} {\bibinfo
  {journal} {Physical Review A}\ }\textbf {\bibinfo {volume} {71}},\ \bibinfo
  {pages} {042315} (\bibinfo {year} {2005})},\ \Eprint
  {http://arxiv.org/abs/arXiv:quant-ph/0408190} {arXiv:quant-ph/0408190}
  \BibitemShut {NoStop}%
\bibitem [{\citenamefont {Ferrie}(2011)}]{Ferrie2011}%
  \BibitemOpen
  \bibfield  {author} {\bibinfo {author} {\bibfnamefont {C.}~\bibnamefont
  {Ferrie}},\ }\href {\doibase 10.1088/0034-4885/74/11/116001} {\bibfield
  {journal} {\bibinfo  {journal} {Reports on Progress in Physics}\ }\textbf
  {\bibinfo {volume} {74}},\ \bibinfo {pages} {116001} (\bibinfo {year}
  {2011})},\ \Eprint {http://arxiv.org/abs/arXiv:1010.2701v3}
  {arXiv:1010.2701v3} \BibitemShut {NoStop}%
\bibitem [{\citenamefont {Fawzi}\ and\ \citenamefont
  {Fawzi}(2018)}]{Fawzi2017a}%
  \BibitemOpen
  \bibfield  {author} {\bibinfo {author} {\bibfnamefont {H.}~\bibnamefont
  {Fawzi}}\ and\ \bibinfo {author} {\bibfnamefont {O.}~\bibnamefont {Fawzi}},\
  }\href {\doibase 10.1088/1751-8121/aab285} {\bibfield  {journal} {\bibinfo
  {journal} {Journal of Physics A: Mathematical and Theoretical}\ }\textbf
  {\bibinfo {volume} {51}},\ \bibinfo {pages} {154003} (\bibinfo {year}
  {2018})},\ \Eprint {http://arxiv.org/abs/1705.06671} {arXiv:1705.06671}
  \BibitemShut {NoStop}%
\bibitem [{\citenamefont {Tomamichel}\ \emph {et~al.}(2009)\citenamefont
  {Tomamichel}, \citenamefont {Colbeck},\ and\ \citenamefont
  {Renner}}]{Tomamichel2009}%
  \BibitemOpen
  \bibfield  {author} {\bibinfo {author} {\bibfnamefont {M.}~\bibnamefont
  {Tomamichel}}, \bibinfo {author} {\bibfnamefont {R.}~\bibnamefont {Colbeck}},
  \ and\ \bibinfo {author} {\bibfnamefont {R.}~\bibnamefont {Renner}},\ }\href
  {\doibase 10.1109/TIT.2009.2032797} {\bibfield  {journal} {\bibinfo
  {journal} {IEEE Transactions on Information Theory}\ }\textbf {\bibinfo
  {volume} {55}},\ \bibinfo {pages} {5840} (\bibinfo {year} {2009})},\ \Eprint
  {http://arxiv.org/abs/0811.1221} {arXiv:0811.1221} \BibitemShut {NoStop}%
\bibitem [{\citenamefont {Mosonyi}\ and\ \citenamefont
  {Hiai}(2011)}]{Mosonyi2011}%
  \BibitemOpen
  \bibfield  {author} {\bibinfo {author} {\bibfnamefont {M.}~\bibnamefont
  {Mosonyi}}\ and\ \bibinfo {author} {\bibfnamefont {F.}~\bibnamefont {Hiai}},\
  }\href {\doibase 10.1109/TIT.2011.2110050} {\bibfield  {journal} {\bibinfo
  {journal} {IEEE Transactions on Information Theory}\ }\textbf {\bibinfo
  {volume} {57}},\ \bibinfo {pages} {2474} (\bibinfo {year}
  {2011})}\BibitemShut {NoStop}%
\bibitem [{\citenamefont {Uhlmann}(1976)}]{Uhlmann1976}%
  \BibitemOpen
  \bibfield  {author} {\bibinfo {author} {\bibfnamefont {A.}~\bibnamefont
  {Uhlmann}},\ }\href {\doibase https://doi.org/10.1016/0034-4877(76)90060-4}
  {\bibfield  {journal} {\bibinfo  {journal} {Reports on Mathematical Physics}\
  }\textbf {\bibinfo {volume} {9}},\ \bibinfo {pages} {273} (\bibinfo {year}
  {1976})}\BibitemShut {NoStop}%
\bibitem [{\citenamefont {Frank}\ and\ \citenamefont {Lieb}(2013)}]{FL13}%
  \BibitemOpen
  \bibfield  {author} {\bibinfo {author} {\bibfnamefont {R.~L.}\ \bibnamefont
  {Frank}}\ and\ \bibinfo {author} {\bibfnamefont {E.~H.}\ \bibnamefont
  {Lieb}},\ }\href {\doibase 10.1063/1.4838835} {\bibfield  {journal} {\bibinfo
   {journal} {Journal of Mathematical Physics}\ }\textbf {\bibinfo {volume}
  {54}},\ \bibinfo {pages} {122201} (\bibinfo {year} {2013})},\ \bibinfo {note}
  {arXiv:1306.5358}\BibitemShut {NoStop}%
\bibitem [{\citenamefont {Beigi}(2013)}]{Beigi2013}%
  \BibitemOpen
  \bibfield  {author} {\bibinfo {author} {\bibfnamefont {S.}~\bibnamefont
  {Beigi}},\ }\href {\doibase 10.1063/1.4838855} {\bibfield  {journal}
  {\bibinfo  {journal} {Journal of Mathematical Physics}\ }\textbf {\bibinfo
  {volume} {54}},\ \bibinfo {pages} {122202} (\bibinfo {year}
  {2013})}\BibitemShut {NoStop}%
\bibitem [{\citenamefont {Wilde}(2018)}]{Wilde2018a}%
  \BibitemOpen
  \bibfield  {author} {\bibinfo {author} {\bibfnamefont {M.~M.}\ \bibnamefont
  {Wilde}},\ }\href {\doibase 10.1088/1751-8121/aad5a1} {\bibfield  {journal}
  {\bibinfo  {journal} {Journal of Physics A}\ }\textbf {\bibinfo {volume}
  {51}},\ \bibinfo {pages} {374002} (\bibinfo {year} {2018})},\ \Eprint
  {http://arxiv.org/abs/arXiv:1710.10252} {arXiv:1710.10252} \BibitemShut
  {NoStop}%
\end{thebibliography}%
\end{document}